\pdfoutput=1 
\newcommand{\LTadd}[1]{}
\newcommand{\LTskip}[1]{#1}

\documentclass[a4paper,12pt]{article}

\usepackage{typearea}
\usepackage[T1]{fontenc}
\usepackage[utf8]{inputenc}
\usepackage[english]{babel}
\usepackage{csquotes}
\usepackage[svgnames]{xcolor}

\usepackage[
    colorlinks,
    allcolors=DarkBlue,
    bookmarksnumbered,
    bookmarksopen,
    bookmarksopenlevel=2,
    bookmarksdepth=3,
    pdfusetitle=true,
    breaklinks=true,
    ]{hyperref}

\newcommand{\statement}[1]{\paragraph{#1}\pdfbookmark[1]{#1}{#1}} 

\usepackage{amsmath, amssymb, amsthm}
\usepackage{mathtools}
\mathtoolsset{centercolon=true}
\DeclareRobustCommand{\colon}{\nobreak\mskip2mu\mathpunct{}\nonscript
  \mkern-\thinmuskip{\ordinarycolon}\mskip6muplus1mu\relax}

\renewcommand{\eqref}[1]{\hyperref[#1]{(\ref*{#1})}}

\usepackage{caption}
\usepackage{floatrow}
\usepackage{microtype}

\allowdisplaybreaks[2] 
\usepackage{tikz}
\usetikzlibrary{shapes.geometric,fit,decorations.markings}
\pgfdeclarelayer{background}
\pgfsetlayers{background,main}
\newlength{\ul}
\newlength{\lw}
\newlength{\vertexdiameter}
\tikzset{%
    myLine/.style={line width=#1\lw,draw=black},
    myLine/.default=1,
    myFace/.style={fill=gray!20},
    x=\ul,
    y=\ul,
    text=black,
    set/.default=gray,
    shorten/.style={shorten <=#1,shorten >=#1},
    shorten/.default=2\lw,
    on layer/.code = {
            \pgfonlayer{#1}\begingroup
            \aftergroup\endpgfonlayer
            \aftergroup\endgroup
    },
    on background/.style = {
        preaction = {%
            #1,
            on layer = background,
        },
    },
    every node/.append style = {
        node font = \sffamily,
        inner sep = 4\lw,
    },
    every set label/.style = {
        draw=none,
        outer sep = \lw,
    },
    set/.style = {
        shape = rectangle,
        color = #1,
        draw,
        on background = {
            fill,
            color = #1!20,
        },
        prefix after command = {
            \pgfextra{\tikzset{every label/.style={color=#1, draw=#1, fill=#1!20, every set label}}}
        }
    }
}
\definecolorset{HTML}{col_}{}{%
    1,0072B2;%
    2,D55E00;%
    3,E69F00
}

\usepackage{enumitem}

\usepackage[
    style=phys,
    biblabel=brackets,
    mincitenames=1,
    maxcitenames=3,
    minalphanames=3,
    maxalphanames=4,
    minbibnames=3,
    maxbibnames=10,
    arxiv=abs,
    eprint=true,
    doi=false,
    url=false,
    sorting=nty,
    useprefix=true,
    alldates=year,
    urldate=long,
]{biblatex}
\renewbibmacro*{note+pages}{%
    \printfield{note}%
    \setunit{\bibpagespunct}%
    \printfield{pages}%
    \newunit%
}
\makeatletter
\DeclareFieldFormat{eprint:arxiv}{%
    \ifhyperref
    {\href{https://arxiv.org/\abx@arxivpath/#1}{%
            arXiv\addcolon\linebreak[2]#1}}
    {arXiv\addcolon
        \nolinkurl{#1}%
    }
}
\makeatother
\DeclareFieldAlias[thesis]{type}{doi/url-link}
\renewbibmacro*{maintitle+title}{%
    \iffieldsequal{maintitle}{title}{%
        \clearfield{maintitle}%
        \clearfield{mainsubtitle}%
        \clearfield{maintitleaddon}%
    }{%
        \iffieldundef{maintitle}{}{%
            \printtext[doi/url-link]{%
                \usebibmacro{maintitle}%
            }%
            \newunit\newblock
            \iffieldundef{volume}{}{%
                \printfield{volume}%
                \printfield{part}%
                \setunit{\addcolon\space}%
            }%
        }%
    }%
    \printtext[doi/url-link]{%
        \usebibmacro{title}%
    }%
    \newunit%
}

\DeclareBibliographyDriver{article}{%
    \usebibmacro{bibindex}%
    \usebibmacro{begentry}%
    \usebibmacro{author/translator+others}%
    \setunit{\labelnamepunct}\newblock
    \usebibmacro{title}%
    \newunit
    \printlist{language}%
    \newunit\newblock
    \usebibmacro{byauthor}%
    \newunit\newblock
    \usebibmacro{bytranslator+others}%
    \newunit\newblock
    \printfield{version}%
    \newunit\newblock
    \printtext[doi/url-link]{%
        \usebibmacro{journal+issuetitle}%
        \newunit
        \usebibmacro{byeditor+others}%
        \newunit
        \usebibmacro{note+pages}%
        \newunit\newblock
        \iftoggle{bbx:isbn}
        {\printfield{issn}}
        {}%
    }%
    \newunit\newblock
    \printfield{pubstate}%
    \setunit{\addspace}%
    \printfield{year}%
    \newunit\newblock
    \usebibmacro{doi+eprint+url}%
    \newunit\newblock
    \printfield{addendum}%
    \setunit{\bibpagerefpunct}\newblock
    \usebibmacro{pageref}%
    \newunit\newblock
    \iftoggle{bbx:related}
        {\usebibmacro{related:init}%
        \usebibmacro{related}}
        {}%
    \usebibmacro{finentry}%
}

\usepackage{xurl} 

\DeclareMathOperator{\sgn}{sgn}

\let\Phi\varPhi

\let\Lambda\varLambda
\let\Lambda\varLambda

\let\Gamma\varGamma

\newcommand{\E}{{\mathrm{e}}}

\newcommand{\I}{\mathrm{i}}
 \newcommand{\R}{ \mathbb{R} }
\newcommand{\C}{ \mathbb{C} }
\newcommand{\N}{ \mathbb{N} }
\newcommand{\Z}{ \mathbb{Z} }
\newcommand{\D}{\mathrm{d}}

\newcommand{\alg}{\mathcal{A}}

\newcommand{\edge}{e}
\newcommand{\unitvec}[1]{{\hat{\mathrm{e}}}_{#1}}

\newcommand{\dmu}{\partial_\mu}
\newcommand{\supedge}{^{\mathrm{edge}}}
\newcommand{\supbulk}{^{\mathrm{bulk}}}
\newcommand{\supany}{^{\mathrm{bulk/edge}}}
\newcommand{\supboundary}{^{\mathrm{boundary}}}

\newcommand{\Iedge}[1][]{I^{\mathrm{edge}}_{#1}}
\newcommand{\tIedge}[1][]{\tilde I^{\mathrm{edge}}_{#1}}
\newcommand{\Iedged}[2][]{I^{#2\,\mathrm{edge}}_{#1}}
\newcommand{\tIedged}[2][]{\tilde I^{#2\,\mathrm{edge}}_{#1}}

\newcommand{\Abulk}{A^{\mathrm{bulk}}}
\newcommand{\Aleft}{A^{\mathrm{left}}}
\newcommand{\Aright}{A^{\mathrm{right}}}
\newcommand{\Abottom}{A^{\mathrm{bottom}}}
\newcommand{\Atop}{A^{\mathrm{top}}}

\newcommand{\Cb}{C_\mathrm{B}}
\newcommand{\Cm}{C_\mathcal{F}}
\newcommand{\CHany}{C_H^{\mkern1mu\mathrm{bulk/edge}}}
\newcommand{\CHbulk}{C_H^{\mkern1mu\mathrm{bulk}}}
\newcommand{\CHedge}{C_H^{\mkern1mu\mathrm{edge}}}

\DeclareMathOperator{\dist}{dist}
\DeclareMathOperator{\diam}{diam}
\DeclareMathOperator{\tr}{tr}

\newcommand{\quadtext}[1]{\quad\text{#1}\quad}

\newcommand{\alignindent}{\hspace{-1cm}}

\newcommand{\sumstack}[2][]{\ifstrempty{#1}{\sum_{\substack{#2}}}{\smashoperator[#1]{\sum_{\substack{#2}}}}}

\usepackage{thm-restate}

\theoremstyle{plain}
\newtheorem{theorem}{Theorem}
\newtheorem{proposition}[theorem]{Proposition}
\newtheorem{corollary}[theorem]{Corollary}
\newtheorem{lemma}[theorem]{Lemma}

\newtheorem{maintheorem}{Theorem}

\theoremstyle{definition}
\newtheorem{definition}[theorem]{Definition}

\theoremstyle{remark}
\newtheorem{remark}[theorem]{Remark}

\title{Equality of magnetization and edge current for interacting lattice fermions at positive temperature}

\date{}

\usepackage{orcidlink}

\makeatletter
\def\blindfootnote{\gdef\@thefnmark{}\@footnotetext}
\makeatother

\author{
    Jonas Lampart%
    \texorpdfstring{%
        \,\orcidlink{0000-0002-6980-3800}
        \footnote{
            \parbox[t]{.7\textwidth}{
                \foreignlanguage{french}{CNRS \& LICB (UMR 6303), Université de Bourgogne,\\
                9 Av.\ A.\ Savary, 21078 Dijon Cedex,} France
            }
        }
    }{}%
    \and Massimo Moscolari%
    \texorpdfstring{%
        \,\orcidlink{0000-0001-7574-1055}
        \footnote{
            \parbox[t]{.7\textwidth}{
                \foreignlanguage{italian}{Dipartimento di Matematica, Politecnico di Milano,\\
                Piazza Leonardo da Vinci~32, 20133~Milano,} Italy
            }
        }
    }{}%
    \and Stefan Teufel%
    \texorpdfstring{%
        \,\orcidlink{0000-0003-3296-4261}
        \footnote{
            \parbox[t]{.7\textwidth}{
                \foreignlanguage{ngerman}{Fachbereich Mathematik, Universität Tübingen,\\
                Auf~der~Morgenstelle~10, 72076~Tübingen,} Germany
            }
        }
    }{}%
    \and Tom Wessel%
    \texorpdfstring{%
        \,\orcidlink{0000-0001-7593-0913}
        \footnotemark[3]
    }{}%
}

\begin{document}

\bgroup

\setlength{\footnotesep}{\dimexpr1.2\baselineskip-\dp\strutbox}

\maketitle

\begin{abstract}
    We prove that the magnetization is equal to the edge current in the thermodynamic limit for a large class of models of lattice fermions with finite-range interactions satisfying local indistinguishability of the Gibbs state, a condition known to hold for sufficiently high temperatures.
    Our result implies that edge currents in such systems are determined by bulk properties and are therefore stable against large perturbations near the boundaries.
    Moreover, the equality persists also after taking the derivative with respect to the chemical potential.
    We show that this form of bulk-edge correspondence is essentially a consequence of homogeneity in the bulk and locality of the Gibbs state.
    An important intermediate result is a new version of Bloch's theorem for two-dimensional systems, stating that persistent currents vanish in the bulk.
\end{abstract}

\thispagestyle{empty}

\blindfootnote{
    Email:\quad%
    \parbox[t]{.7\textwidth}{
        \hypersetup{hidelinks}
        \href{mailto:jonas.lampart@u-bourgonge.fr}{jonas.lampart@u-bourgonge.fr},
        \href{mailto:massimo.moscolari@polimi.it}{massimo.moscolari@polimi.it},
        \newline
        \href{mailto:stefan.teufel@uni-tuebingen.de}{stefan.teufel@uni-tuebingen.de},
        \href{mailto:tom.wessel@uni-tuebingen.de}{tom.wessel@uni-tuebingen.de}
    }
}

\blindfootnote{%
    \parbox[t]{\linewidth-\parindent}{
        A prior version of the article has been accepted for publication, after peer review.
        This version includes some post-acceptance changes but is not the Version of Record and does not reflect post-acceptance improvements, or any corrections.
        The Version of Record is available online at
        \href{http://dx.doi.org/10.1007/s11040-024-09495-8}{doi:\,10.1007/s11040-024-09495-8}.
    }
}

\egroup

\newpage

\tableofcontents

\section{Introduction}

We show that extended fermion systems subject to homogeneous magnetic fields exhibit a form of bulk-edge correspondence in the thermodynamic limit, namely exact equality of magnetization and edge current, at positive temperatures.
Roughly speaking, our assumptions are finite-range interactions, homogeneity in the bulk of the Hamiltonian, and local indistinguishability of the Gibbs state.
The first two are explicit assumptions on the Hamiltonian, the last is known to hold for sufficiently high temperatures~\cite{KGK2014} and expected to hold much more generally.

In~\cite{CMT2021} a similar result was established for non-interacting fermion systems.
There it is also shown, how this result relates to the better known bulk-edge correspondence of the transport coefficients: Under the assumption of a gapped ground state and in the zero temperature limit, the derivative of the magnetization with respect to the chemical potential converges to the Hall conductivity and the derivative of the edge current with respect to the chemical potential converges to the edge conductance.
In this paper we establish the differentiability of the magnetization with respect to the chemical potential also for interacting systems and thus also the equality of the corresponding derivatives.

Let us now be more specific.
We consider a system of interacting fermions modelled by a sequence of finite-range Hamiltonians \((H_L(b))_L\) defined on boxes \(\Lambda_L=\{-L,\dotsc,L\} \times \{0,\dotsc,2L\}\) and dependent on a homogeneous magnetic field \(b\) perpendicular to the plane.
We think of \(\Lambda_L\) as a subset of the upper half plane of \(\Z^2\) and consider a strip \(\{-L,\dotsc,L\} \times \{0,\dotsc,D-1\}\) of fixed width \(D\) as the edge region and its complement as the bulk.
In the bulk we assume translation invariance of the Hamiltonian with respect to magnetic translations.

For inverse temperature \(\beta>0\), chemical potential \(\mu \in \R\), and magnetic field \(b\in\R\) the Gibbs, or thermal, state is defined as
\begin{equation}
    \label{eq:Gibbs}
    \rho_L(\beta,\mu,b)
    :=
    \frac{\E^{-\beta (H_L(b) - \mu \,\mathcal{N}_L)}}{\mathcal{Z}_L(\beta,\mu,b)}
    ,
\end{equation}
where \(\mathcal{N}_L\) is the number operator and \(\mathcal{Z}_L(\beta,\mu,b):=\tr\bigl(\E^{-\beta (H_L(b) - \mu \,\mathcal{N}_L)}\bigr)\) is the partition function.
In the absence of interactions, \(\rho_L(\beta,\mu,b)\) is naturally a local object, namely it has an integral kernel in which it is possible to identify a bulk and an edge region, see e.g.~\cite{CMT2021}.
However, in the interacting setting, the locality of \(\rho_L(\beta,\mu,b)\) is a delicate issue, which has been investigated e.g.\ in~\cite{DFF1996a,DFF1996,Yarotsky2005,KGK2014,FU2015}, see also the more recent~\cite{CMTW2023}.
In the present work locality of the Gibbs state in the form of \emph{local indistinguishability} is one of the crucial assumptions: Let \(X\subset\Lambda'\subset \Lambda_L\), then we assume that the expectation value of an observable \(A \in \alg_X\) can be approximated by the Gibbs state of the Hamiltonian restricted to \(\Lambda'\),
\begin{equation*}
    \tr\bigl( \rho_L(\beta,\mu,b) \, A \bigr)
    \approx
    \tr\bigl( \rho_{\Lambda'}(\beta,\mu,b) \, A \bigr),
\end{equation*}
up to terms that vanish in the distance of \(X\) to the boundary \(\partial\Lambda'\) of \(\Lambda'\).
A subtle point here is that the definition of \(\partial\Lambda'\) depends on whether we consider \(\Lambda'\) as a subset of \(\Z^2\) or as a subset \(\Lambda_L\).
In the first case \(\partial\Lambda'\) could include parts of the physical edge \(\{-L,\ldots,L\}\times \{0\}\) and local indistinguishability is only demanded for \(X\) in the bulk of the system.
In the second case local indistinguishability is also required for \(X\) located at the edge of the system.
For this reason we speak of \emph{local indistinguishability in the bulk} for the former case and \emph{local indistinguishability everywhere} for the latter.
Note that in our setting a sufficient condition implying local indistinguishability everywhere is a sufficiently high temperature~\cite{KGK2014}.
However, for systems with short-range interactions, one may generally expect local indistinguishability to hold away from critical points, i.e.~whenever the system has a unique thermal state in the thermodynamic limit.
Such a state has decaying correlations~\cite[chapter~4]{Simon1993}, which implies local indistinguishability, at least if the decay is sufficiently fast~\cite{CMTW2023}.

The magnetization is defined as the derivative of the grand canonical pressure \(p_L(\beta,\mu,b) := - \lvert\Lambda_L\rvert^{-1} \, \beta^{-1} \ln \bigl(\mathcal{Z}_L(\beta,\mu,b)\bigr)\) with respect to the magnetic field~\(b\), namely
\begin{equation}
    \label{eq:magnetization}
    m_L(\beta,\mu,b) := \frac{\partial}{\partial b} \, p_L(\beta,\mu,b)
    .
\end{equation}

Our main result states that whenever the family of finite volume Gibbs states satisfies \emph{local indistinguishability in the bulk} then the magnetization approximately equals the bond current \(\Iedge[L](\beta,\mu,b)\) through an orthogonal line of length~\(L\) at the lower edge of the sample, see Figure~\ref{fig:results} and equation~\eqref{eq:edge-current},
\begin{equation}
    \label{eq:introBEC}
    \bigl\lvert
        m_L(\beta,\mu,b) - \Iedge[L](\beta,\mu,b)
    \bigr\rvert
    =
    \mathcal{O}(L^{-1}).
\end{equation}
Moreover, this current is very well localized near the edge and thus called \emph{edge current}.
Both statements are contained in Theorem~\ref{thm:main} and depicted in Figure~\ref{fig:results}.
Furthermore, if the finite volume Gibbs state satisfies \emph{local indistinguishability everywhere}, then we show that the thermodynamic limit \(m(\beta,\mu,b):=\lim_{L\to\infty} m_L(\beta,\mu,b)\) exists and obtain an equality between the edge current \(\Iedge(\beta,\mu,b)\) and the magnetization \(m(\beta,\mu,b)\) in the infinite volume system, see Theorem~\ref{thm:limit-I}\@.
And while the orbital magnetization \(m_L(\beta,\mu,b)\) and the edge current \(\Iedge[L](\beta,\mu,b)\) of the finite systems in principle depend on the bulk and edge part of the system, we show that the limits \(m(\beta,\mu,b)\) and \(\Iedge(\beta,\mu,b)\) are independent of the specific shape of the interactions at the edge of the system.

However, since we obtain the infinite volume magnetization \(m(\beta,\mu,b)\) from a limit of finite systems with edges converging to a system on the upper half plane with an infinite edge, one might ask whether \(m(\beta,\mu,b)\) can be considered a pure bulk quantity.
To answer this question in the affirmative, we show in Theorem~\ref{thm:limit-p} that the magnetization obtained from any KMS state at \((\beta,\mu)\) for the translation invariant bulk Hamiltonian defined on the entire plane coincides with \(m(\beta,\mu,b)\).

Finally, in Theorem~\ref{thm:muderivative} we establish the differentiability of \(m(\beta,\mu,b)\) with respect to~\(\mu\).
By comparison to the non-interacting setting, one would expect that, in the presence of a spectral gap and with weak interactions, the zero temperature limit of \(\dmu m(\beta,\mu,b)\) converges to the quantized Hall conductivity.
While this result is not present in the literature and out of the scope of the paper, we show here a preliminary regularity result of \(m(\beta,\mu,b)\) with respect to~\(\mu\).

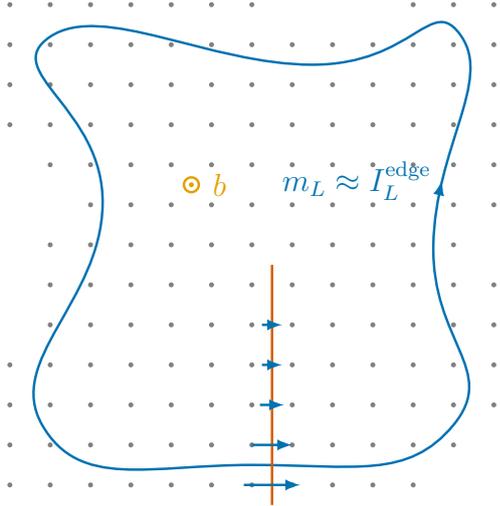
\begin{figure}
    \fcapside[.42\textwidth]{
        \setlength{\vertexdiameter}{2\lw}
        \setlength{\ul}{(.42\textwidth-\vertexdiameter-1ex)/12}
        \begin{tikzpicture}[myLine,x=\ul,y=\ul]
            \path[save path=\shape] (-6.2,12.2) -- ++ (0.0,-3.7) -- ++ (1.4,-1.0) -- ++ (0.0,-1.0) -- ++ (-0.7,-0.5) -- ++ (0.0,-2.0) -- ++ (-0.7,-0.5) -- (-6.2,-0.2) -- (4.5,-0.2) -- ++ (1.7,1.7) -- ++ (0.0,4.0) -- + (-0.7,0.5) -- ++ (0.0,1.0) -- (6.2,12.2) -- ++ (-2.7,0.0) -- ++ (-0.5,-0.7) -- ++ (-2.0,0.0) -- ++ (-0.5,0.7) -- cycle;
            \begin{scope}
                \clip[use path=\shape];
                \foreach \x in {-6,...,6}{
                    \foreach \y in {0,...,12}{
                        \node[style=circle, inner sep=0pt, outer sep=0, minimum width=\vertexdiameter, fill=gray,
                        ] (\x;\y) at (\x,\y) {};
                    }
                }
            \end{scope}
            \begin{scope}
                \clip(current bounding box.south west) rectangle (current bounding box.north east);
                \draw[
                    color=col_1,
                    line width=\lw,
                    decoration={
                        markings,
                        mark=at position 0.15 with {
                            \arrow{latex}
                            \node[left] {\(m_L \approx \Iedge[L]\)};
                        }
                    },
                    postaction=decorate
                ] plot[smooth cycle, tension=.9] coordinates {(0.0,0.5) (5.0,1.5) (4.5,6.0) (5.2,11.2) (1.5,10.5) (-5.0,11.2) (-3.7,7.0) (-5.2,1.5)};
                \path[every node/.append style={style=circle, inner sep=0pt, color=col_3}] (-1.5,7.5)
                    node[minimum width=2\lw, fill] {}
                    node[minimum width=6\lw, draw] {}
                    node[right,inner sep=4\lw] {\(b\)};
            \end{scope}
            \begin{scope}[draw=col_1]
                \path[draw=col_2] (0.5,-0.5) -- (0.5,5.5);
                \foreach \y / \linelen in {0/.7,1/.5,2/.3,3/.25,4/.25}{
                    \draw[-latex] (0.5,\y) +(-\linelen,0) -- +(\linelen,0);
                }
            \end{scope}
        \end{tikzpicture}
    }{
        \caption{
            Pictorial representation of our main results:
            For locally interacting fermions on a two-dimensional lattice with perpendicular magnetic field \(b\), satisfying local indistinguishability (see details in section~\ref{sec:Hamiltonian}) at positive temperature, the edge current \(\Iedge[L]\), which is the bond current through the vertical line, is localized near the boundary and approximately equals the magnetization \(m_L\) (Theorem~\ref{thm:main}\@).
            The latter is a bulk quantity, i.e.\ it converges in the thermodynamic limit \(L\to\infty\) and does not depend on the details near the boundary (Theorem~\ref{thm:limit-I}\@).
            This independence allows for the rough edges in the picture.
        }
        \label{fig:results}
    }
\end{figure}

A crucial ingredient to our proofs is a new version of \nameref{prop:Bloch-theorem} for two-dimensional systems.
We show that local indistinguishability together with current conservation implies that currents decay quickly with the distance to the edge, or, put differently, that in equilibrium currents can only flow near the edge of a sample.
See~\cite{Watanabe2019,BF2021} for recent related results and the discussion below.

Let us end the introduction with a few more comments on the literature.
As already mentioned, analogous mathematical results relating magnetization in the bulk with edge currents for non-interacting systems were obtained in~\cite{CMT2021}, with predecessors e.g.\ in~\cite{ANB1975,MMP1988,Kunz1994}.
Notice that the equality between edge current and magnetization in two-dimensional systems can be also interpreted as a quantum mechanical, microscopic version of Ampère’s law, as it is sometimes addressed in the physics literature, see for instance~\cite{NRLB2017}, where the effect of a time-dependent magnetic field on the magnetization of localized states is analyzed in a discrete, non-interacting setting.
The existence and properties of edge states of magnetic Schrödinger operators were studied e.g.\ in~\cite{FGW2000,DP2002}.
The mathematical literature on bulk-edge correspondence for transport coefficients is vast but concerns almost exclusively non-interacting systems at zero temperature and with a gap in the bulk, e.g.~\cite{SKR2000,EGS2005}.
In~\cite{Froehlich2018,FK1991,FS1993}
the authors derive, starting from the assumption of an incompressible bulk, effective actions for the bulk and the edge system.
While they do not start from a many-body fermion model as we do, they are able to derive much more far-reaching consequences for quantum Hall systems from a seemingly innocuous assumption about the bulk.
In microscopic models of interacting fermions the bulk-edge correspondence of transport coefficients was established at zero-temperature for weakly interacting gapped systems in~\cite{GMP2017,MP2022}.

\section{Mathematical framework and main results}
\label{sec:results}

\subsection{The Hamiltonian}
\label{sec:Hamiltonian}

Let \(\Z_+ = \Z\cap [0,\infty)\) and \(\Z^2_+ = \Z \times \Z_+\), both equipped with the \(1\)-metric \(\dist(x,y) := \lvert x_1-y_1 \rvert + \lvert x_2-y_2 \rvert\).
For any finite subset \(X\Subset\Z^2\) let \(\mathfrak{h}_X:= \ell^2(X, \C^s)\) be the one-body space and \(\mathfrak{F}_X:= \mathfrak{F}^- (\mathfrak{h}_X)\) the corresponding fermionic Fock space.
By \(\alg_X\) we denote the algebra of all bounded operators in \(\mathcal{L}(\mathfrak{F}_X)\) that commute with the number operator \(\mathcal{N}_X:= \sum_{x\in X}a^*_x \, a^{}_x:=\sum_{x\in X}\sum_{j=1}^s a^*_{x,j} \, a^{}_{x,j}\) and by \(\alg_\mathrm{loc} := \bigcup_{X\Subset\Z^2} \alg_X\) the algebra of all local observables that preserve particle number.
Its closure
\begin{equation*}
    \alg := \overline{\alg_\mathrm{loc}}^{\lVert\cdot\rVert}
\end{equation*}
is a \(C^*\)-algebra and called the quasi-local algebra.

We consider sequences \((H_L(b))_{L\in\N}\) of Hamiltonians defined on boxes \(\Lambda_L := \bigl( [-L,L] \times [0,2L] \bigr) \cap \Z^2\) that are of the form
\begin{equation}
    \label{eq:Hb}
    \begin{aligned}
        H_L(b)
        &=
        \begin{aligned}[t]
            &
            \sum_{x,y \in \Lambda_L} a^*_x \, T\supbulk_b(x, y) \, a^{}_y
            + \sum_{X\subset \Lambda_L} \Phi\supbulk(X)
            \\&
            +\sum_{x,y \in \Lambda_L} a^*_x \, T\supedge_b(x, y) \, a^{}_y
            + \sum_{X\subset \Lambda_L} \Phi\supedge(X)
        \end{aligned}
        \\&=:
        \sum_{x,y \in \Lambda_L} a^*_x \, T_b(x, y) \, a^{}_y
        + \sum_{X\subset \Lambda_L} \Phi(X)
        .
    \end{aligned}
\end{equation}
We sometimes need to restrict this Hamiltonian to other finite sets \(\Lambda'\subset \Lambda_L\).
In this case we write \(H_L(b)\big|_{\Lambda'}\) which means that the sums in~\eqref{eq:Hb} only run over \(x,y\in \Lambda'\) and \(X\subset \Lambda'\), respectively.

The Hamiltonian is split into a “bulk” part, which is invariant under magnetic translations and an “edge” part, which lives on the lower edge.
Each consists of two contributions:
The kinetic terms
\begin{equation}
    \label{eq:definition-Tb}
    T\supany_b(x,y)
    := \E^{ \I \frac{x_2+y_2 \xmathstrut{.2}}{2} b (x_1-y_1)} \, T\supany(x,y)
\end{equation}
are a Peierls phase times a hopping amplitude \(T\supany\colon \Z^2\times\Z^2 \to \mathcal{L}(\C^s)\), which is uniformly bounded \(\sup_{x,y\in \Z^2} \bigl\lVert T\supany(x,y) \bigr\rVert \leq C\) and satisfies \(T\supany(x,y) = {T\supany(y,x)}^*\).
The interactions
\begin{equation*}
    \Phi\supany\colon \{X\Subset\Z^2\}\to \alg_\mathrm{loc}
    ,\quad
    X\mapsto \Phi\supany(X)\in\alg_X
\end{equation*}
are self-adjoint, the terms are uniformly bounded, \(\sup_{X\subset \Z^2}\lVert\Phi\supany(X)\rVert \leq C\), and the corresponding operators \(\sum_{X\subset\Lambda_L} \Phi\supany(X)\) are assumed to commute with all local number operators \(\mathcal{N}_{\{z\}}\) for \(z\in \Lambda_L\).
The last condition is satisfied, e.g., for density-density interactions or external potentials.

Furthermore, all terms are assumed to be of finite range \(R\in \N\), i.e.\ \(T\supany(x,y)=0\) if \(\dist(x,y) > R\) and \(\Phi\supany(X)=0\) if \(\diam(X) > R\).
As mentioned above, the bulk contributions are assumed to be invariant under magnetic translations, i.e.\ \(T\supbulk(x-z,y-z)=T\supbulk(x,y)\) only depends on the difference \(x-y\) and \(\Phi\supbulk\) satisfies~\eqref{eq:invariance-under-magnetic-translations}.
And the edge contributions are supported on a strip of fixed width \(D\) along the lower edge, i.e.\ \(T\supedge(x,y)=0\) unless \(x,y\in \Z\times\{0,1,\dotsc,D-1\}\) and \(\Phi\supedge(X)=0\) unless \(X\subset \Z\times\{0,1,\dotsc,D-1\}\).
Without loss of generality we choose \(D \geq R\), since the presence of the boundary already modifies the Hamiltonian in \(\Z\times\{0,1,\dotsc,R-1\}\).

A canonical example of a magnetic Hamiltonian with interactions is the Hofstadter-Hubbard model, i.e.\ the second quantization of the discrete magnetic Laplacian together with an on-site density-density interaction.
More precisely, for the Hofstadter-Hubbard model we have
\(\mathfrak{h}_{\{x\}}=\C^2\),
\begin{equation*}
    T\supbulk(x, y) = {\mathrm{id}_{\C^2} \cdot\delta_{\lvert x-y \rvert=1}},
    \quadtext{and}
    \Phi\supbulk(X) = {a^*_{x,1} \, a^{}_{x,1} \, a^*_{x,2} \, a^{}_{x,2} \cdot \delta_{X=\{x\}}}
    ,
\end{equation*}
which leads to
\begin{equation*}
    H_L^\mathrm{HH}(b)
    =
    \sumstack[lr]{x,y\in\Lambda_L:\\\lvert x-y \rvert=1}
    \,\E^{ \I \frac{x_2+y_2 \xmathstrut{.2}}{2} b (x_1-y_1)}
    \sumstack[lr]{j\in \{1,2\}} a^*_{x,j} \, a^{}_{y,j}
    + \sum_{x\in\Lambda_L} a^*_{x,1} \, a^{}_{x,1} \, a^*_{x,2} \, a^{}_{x,2}
    .
\end{equation*}
Near the edge, one could, for example, add an external potential \(\Phi\supedge(X) = \phi(x) \, \mathcal{N}_{\{x\}} \, \delta_{X=\{x\}}\) or effectively remove individual sites by subtracting all hoppings connected to them.

For a finite subset \(\Lambda\Subset\Z^2_+\), a Hamiltonian \(H\in \alg_\Lambda\), inverse temperature \(\beta>0\), chemical potential \(\mu \in \R\), and magnetic field \(b\in\R\) we denote the grand canonical partition function by
\begin{equation*}
    \mathcal{Z}_\Lambda[H](\beta,\mu) := \tr \bigl(\E^{-\beta (H - \mu \,\mathcal{N}_\Lambda)} \bigr)
    ,
\end{equation*}
and the Gibbs state by
\begin{equation*}
    \rho_{\Lambda}[H](\beta,\mu) := \frac{\E^{-\beta (H - \mu \,\mathcal{N}_\Lambda)}}{\mathcal{Z}_\Lambda[H](\beta,\mu)}
    .
\end{equation*}
When we consider truncated Hamiltonians in the proofs, we drop the index and informally write \(\rho\bigl[H_L(b)\big|_{\Lambda'}\bigr](\beta,\mu) \equiv \rho_{\Lambda'}\bigl[H_L(b)\big|_{\Lambda'}\bigr](\beta,\mu)\).
On the boxes \(\Lambda_L\) we abbreviate
\begin{equation*}
    \mathcal{N}_L := \mathcal{N}_{\Lambda_L}
    ,\quad
    \mathcal{Z}_L(\beta,\mu,b) := \mathcal{Z}_{\Lambda_L}\bigl[H_L(b)\bigr](\beta,\mu)
    \quadtext{and}
    \rho_L(\beta,\mu,b) := \rho_{\Lambda_L}\bigl[H_L(b)\bigr](\beta,\mu)
    .
\end{equation*}

The key hypothesis for our results is that the Gibbs state is locally determined by the local terms in the Hamiltonian.
This property is often called \emph{local indistinguishability} and made precise in the following definition.

\begin{definition}[Local indistinguishability of the Gibbs state]\label{def:local-indistinguishability}
    Let \(\zeta\colon\N_0\to \R_+\) be non-increasing with \(\lim_{n\to \infty}\zeta(n)=0\) and \(g\colon\R_+\to \R_+\) non-decreasing with \(g\bigl((2R+1)^2+1\bigr)=1\).
    The family of Hamiltonians \((H_L(b))_{L\in\N}\) is said to satisfy \emph{local indistinguishability of the Gibbs state at \((\beta,\mu,b)\)} with \(\zeta\)-decay,
    \begin{enumerate}[label=(\alph*)]
        \item \emph{in the bulk} if and only if for all \(L\in \N\), \(X\subset \Lambda'\subset \Lambda_L\) and \(A\in \mathcal{L}(\mathcal{F}_X) \subset \alg\)
            \begin{equation}\label{eq:def-local-indistinguishability-bulk}
                \bigl\lvert
                    \tr \bigl(\rho_L(\beta,\mu,b) \, A\bigr)
                    - \tr \bigl(\rho_{\Lambda'}\bigl[H_L(b)\big|_{\Lambda'}\bigr](\beta,\mu) \, A\bigr)
                \bigr\rvert
                \leq
                \lVert A \rVert
                \, g\bigl( \lvert X \rvert \bigr)
                \, \zeta\bigl( \dist(X,\Z^2\setminus\Lambda')\bigr),
            \end{equation}

        \item \emph{everywhere} if and only if for all \(L\in \N\), \(X\subset \Lambda'\subset \Lambda_L\) and \(A\in \mathcal{L}(\mathcal{F}_X) \subset \alg\)
            \begin{equation}\label{eq:def-local-indistinguishability-everywhere}
                \bigl\lvert
                    \tr \bigl(\rho_L(\beta,\mu,b) \, A\bigr)
                    - \tr \bigl(\rho_{\Lambda'}\bigl[H_L(b)\big|_{\Lambda'}\bigr](\beta,\mu) \, A\bigr)
                \bigr\rvert
                \leq
                \lVert A \rVert
                \, g\bigl( \lvert X \rvert \bigr)
                \, \zeta\bigl( \dist(X,\Lambda_L \setminus \Lambda')\bigr)
                .
            \end{equation}
    \end{enumerate}
\end{definition}

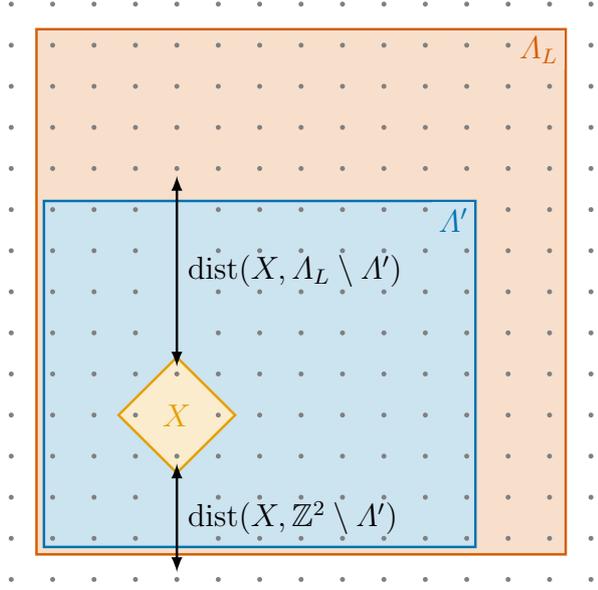
\begin{figure}
    \fcapside[.55\textwidth]{
        \setlength{\vertexdiameter}{2\lw}
        \setlength{\ul}{(.5\textwidth-\vertexdiameter-1ex)/14}
        \begin{tikzpicture}[myLine,x=\ul,y=\ul]
            \foreach \x in {-7,...,7}{
                \foreach \y in {-1,...,13}{
                    \node[style=circle, inner sep=0pt, outer sep=0pt, minimum width=\vertexdiameter, fill=gray] (\x;\y) at (\x,\y) {};
                }
            }
            \begin{scope}[every node/.append style={inner sep=2.5\lw}]
                \node[set=col_2, inner sep=5.5\lw, fit={(-6;0) (6;12)},
                    label={[anchor=north east]north east:\(\Lambda_L\)}] (Lambda L) {};
                \node[set=col_1, fit={(-6;0) (4;8)}, save path=\LambdaPrime,
                    label={[anchor=north east]north east:\(\Lambda'\)}] (Lambda Prime) {};
                \node[set=col_3, diamond, fit={(-4;3) (-2;3)},
                    label={[anchor=center]center:\(X\)}] (set X) {};
            \end{scope}
            \begin{scope}[draw, shorten, latex-latex, every edge/.append style={preaction={draw, line width=3\lw, arrows=-, shorten=\ul-3\lw, color=col_1!20}}]
                \path (-3;2) edge node[right] {\(\dist(X,\Z^2\setminus\Lambda')\)} (-3;-1);
                \path (-3;4) edge node[right] {\(\dist(X,\Lambda_L\setminus\Lambda')\)} (-3;9);
            \end{scope}
        \end{tikzpicture}
    }{
        \caption{
            Sketch of the two distances used in Definition~\ref{def:local-indistinguishability}.
            If the Hamiltonian satisfies local indistinguishability in the bulk, the bound~\eqref{eq:def-local-indistinguishability-bulk} decays in the distance to the boundary of \(\Lambda'\) (in \(\Z^2\)).
            While local indistinguishability everywhere also gives a good estimate if \(X\) is close to the boundary of \(\Lambda'\) as long as the boundaries of \(\Lambda'\) and \(\Lambda_L\) coincide in that region.
        }
        \label{fig:distances-local-indistinguishability}
    }
\end{figure}

Note the difference between \(\Z^2\setminus \Lambda'\) and \(\Lambda_L\setminus \Lambda'\) in the distance in~\eqref{eq:def-local-indistinguishability-bulk} and~\eqref{eq:def-local-indistinguishability-everywhere}, if \(\Lambda'\) includes parts at the boundary of \(\Lambda_L\), see Figure~\ref{fig:distances-local-indistinguishability}.
In particular \(\dist(X,\Z^2\setminus\Lambda') = \min \{\dist(X,\Lambda_L\setminus{\Lambda'}), \dist(X,\Z^2\setminus\Lambda_L)\}\), so indistinguishability everywhere implies the property in the bulk.
This distinction is useful, because we expect a better decay in the bulk and a worse decay at the boundary due to the presence of edge states.
For some of the statements we however need local indistinguishability also near the boundary and might accept a slower decay.
In particular, local indistinguishability directly implies decay of correlations, see Lemma~\ref{lem:decay-of-correlations}, and we do not expect that to hold with good decay near the boundary due to edge states, see e.g.~\cite{MP2022}.

For most of our results, we will require local indistinguishability with decay at least \(\zeta\in \ell^1\), but any better decay will yield better results, in particular concerning localization near the boundary.
For example, local indistinguishability everywhere with exponential decay function \(\zeta\) is known to hold for sufficiently high temperature in systems with finite-range interactions~\cite[Corollary~2 and~5]{KGK2014}.
The decay rate and constants depend on \(\beta\), but can be chosen uniformly for small~\(\beta\).
As is shown in~\cite{CMTW2023}, decay of correlations at some positive temperature implies local indistinguishability at the same temperature, and the converse also holds, see Lemma~\ref{lem:decay-of-correlations}.

The normalization in Definition~\ref{def:local-indistinguishability} is chosen such that all later bounds, where we always restrict to sets \(\lvert X \rvert \leq (2R+1)^2 + 1\) so that \(g\bigl(\lvert X \rvert\bigr)\leq 1\), do not depend on \(g\).
We need to allow for larger \(X\) only to define the thermodynamic limit \(\rho_\infty(\beta,\mu,b)\).

\subsection{The edge current and the magnetization}
Denote by \(B_L^x(\ell) := \{\, y\in \Lambda_L \,\vert\, \dist(x,y) \leq \ell \,\}\) the ball around \(x\) in \(\Lambda_L\) with radius \(\ell\).
The set \(B_L^x(R)\) contains all points which have non-vanishing interaction with \(x\).
Then the current operator \(J_L(b)\) has components (\(k=1,2\))
\begin{align*}
    J_{k,L}(b)
    & := \I \, \bigl[X_{k,L}, H_L(b)\bigr]
    = \I \Biggl[
        \sum_{z\in \Lambda_L} z_k \, a^*_z \, a^{}_z ,
        \sum_{x,y \in \Lambda_L} a_x^* \, T_b(x,y) \, a^{}_y
    \Biggr]
    \\ &\LTadd{=} \mathrel{\phantom{:=}\llap{\(=\)}} \I \sum_{x \in \Lambda_L} \sum_{y \in B_L^x(R)} (x_k-y_k) \, a_x^*\, T_b (x,y) \,a^{}_y.
\end{align*}
We now rewrite this sum as a sum of currents through edges of the dual lattice.
For that, denote by \(\edge_{k,z}\subset \R^2\) the dual edge which intersects the edge between lattice points \(z\) and \(z+\unitvec{k}\) and by \(\overline{\edge_{k,z}}\) the edge together with the attached vertices (see Figure~\ref{fig:dual-lattice-and-current-operator}).
Here, \(\unitvec{k}\) denotes the unit vector in \(k\) direction, e.g.\ \(\unitvec{1}=(1,0)\).
Moreover, denote by \(\overline{xy}\subset \R^2\) the line connecting \(x\) and \(y\).
We define the current through the dual edge \(\edge_{k,z}\) as
\begin{equation}
    \label{eq:current-through-dual}
    J_{k,L}^z(b)
    := \frac{\I}{2}\biggl(
        \sumstack[r]{x,y\in \Lambda_L:\\\overline{xy}\,\cap\, \edge_{k,z} \neq \emptyset} \sgn(x_k-y_k) \, a_x^*\, T_b (x,y) \,a^{}_y
        + \sumstack[r]{x,y\in \Lambda_L:\\\overline{xy}\,\cap\, \overline{\edge_{k,z}} \neq \emptyset} \sgn(x_k-y_k) \, a_x^*\, T_b (x,y) \,a^{}_y
    \biggr)
    .
\end{equation}
Thus, each hopping term \(a_x^*\, T_b (x,y) \,a^{}_y\) is included once in \(J_{k,L}^z(b)\) if \(\overline{xy}\) intersects the dual edge \(e_{k,z}\), or half if \(\overline{xy}\) intersects only the endpoints \(\overline{e_{k,z}}\setminus e_{k,z}\) of the dual edge.
In the latter case it appears for twice as many different~\(z\).
Since \(\overline{xy}\) intersects \(\lvert x_k - y_k \rvert\) vertical lines we can rewrite
\begin{equation*}
    J_{1,L}(b)
    = \sum_{m=-L}^{L-1} \sum_{n=0}^{2L} J_{1,L}^{(m,n)}(b)
    \quadtext{and}
    J_{2,L}(b)
    = \sum_{m=-L}^{L} \sum_{n=0}^{2L-1} J_{2,L}^{(m,n)}(b)
\end{equation*}
by summing over all edges.
Note, that the \(L\)-dependence of \(J_{k,L}^z(b)\) only stems from missing hopping terms near the boundary, and we define
\begin{equation*}
    J_{k}^z(b) := J_{k,L}^z(b)
\end{equation*}
for all \(L > \lvert z_1 \rvert + R\) and \(2L > z_2 +R\) consistently.

\begin{figure}
    \fcapside{
        \setlength{\ul}{1cm}
        \setlength{\vertexdiameter}{3\lw}
        \colorlet{color_int}{col_1}
        \colorlet{color_point}{col_2}
        \colorlet{color_edge}{col_2}
        \begin{tikzpicture}[myLine,x=\ul,y=\ul, line cap=round]
            \coordinate (z) at (0,0);
            \path (z) +(0.5,0) coordinate (center);
            \draw[color=color_int] (center) +(-.5,0) -- +(.5,0);
            \foreach \mx in {-1,1}{
                \draw[color=color_int] (center) [yshift=\mx*2\lw] +(-\mx*0.5,0) -- +(\mx*1.5,0);
                \foreach \my in {-1,1}{
                    \draw[color=color_int!60] (center) +(-\mx*0.5,0) -- +(\mx*0.5,\my);
                }
            }
            \foreach \x in {-1,-0,...,4}{
                \foreach \y in {-2,-1,...,2}{
                    \path[fill=gray] (\x,\y) circle (\vertexdiameter);
                }
            }
            \fill[color=color_point] (z) circle (\vertexdiameter) node[anchor=south east] {\(z\)};
            \draw[color=color_edge, line cap=round, line width=2\lw] (center) +(0,-.5) node[anchor=west] {\(e_{1,z}\)} -- +(0,.5);
            \fill[color=color_point] (3,1) coordinate (z') circle (\vertexdiameter) node[anchor=north west] {\(z'\)};
            \path (z') +(0,0.5) coordinate (center');
            \draw[color=color_edge, line cap=round, line width=2\lw] (center') +(.5,0) node[anchor=west] {\(e_{2,z'}\)} -- +(-.5,0);
        \end{tikzpicture}
    }{
        \caption{
            The figure shows a small section of \(\Lambda_L\) and two points \(z\) and \(z' \in \Lambda_L\) with their dual edges \(e_{1,z}\) and \(e_{2,z'}\), respectively. 
            For \(z\), also all lines \(\overline{xy}\) which contribute to \(J_{1,L}^z(b)\) for \(R=2\) are drawn.
            The light blue lines only intersect the endpoints of \(\overline{e_{1,z}}\) and thus come with a prefactor \(1/2\) in~\eqref{eq:current-through-dual}. 
        }
        \label{fig:dual-lattice-and-current-operator}
    }
\end{figure}
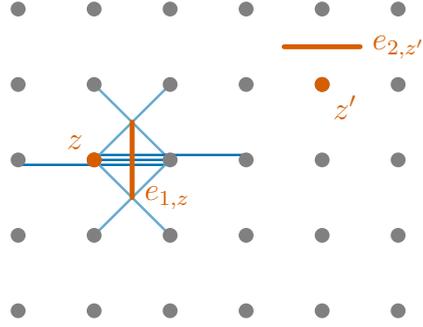

Moreover, for \(d\in \{1,\dotsc,L\}\), we define the edge current as
\begin{equation}
    \label{eq:edge-current}
    \Iedged[L]{d}(\beta,\mu,b)
    := \sum_{n=0}^{d-1} \tr \bigl( \rho_L(\beta,\mu,b) \, J_{1,L}^{(0,n)}(b)\bigr)
\end{equation}
and introduce the shorthand \(\Iedge[L]=\Iedged[L]{L}\) for the current along the lower half of the system \(\Lambda_L\).
By current conservation it equals the currents through lines connecting the center of \(\Lambda_L\) with the midpoints of the other boundaries of \(\Lambda_L\), see proof of Proposition~\ref{prop:m_L-equals-I_L^d}, we only choose this edge because it persists in the thermodynamic limit \(L\to\infty\) in our geometry.

It remains to recall the definition of (orbital) magnetization.
For inverse temperature \(\beta>0\), chemical potential \(\mu \in \R\), and magnetic field \(b\in\R\), the grand canonical pressure is given by
\begin{equation}
    \label{eq:definition-pressure-finite-volume-partition-function}
    p_L(\beta,\mu,b)
    :=
    -(2L+1)^{-2}\beta^{-1} \ln \bigl(\mathcal{Z}_L(\beta,\mu,b)\bigr)
    ,
\end{equation}
and the magnetization by
\begin{equation*}
    m_L(\beta,\mu,b)
    :=
    \frac{\partial}{\partial b} \, p_L(\beta,\mu,b)
    .
\end{equation*}

\subsection{Main results}
\label{subsec:results}
Our first main result deals with the magnetization and the edge current at finite volume.
For this type of result only local indistinguishability of the Gibbs state in the bulk is needed.

\belowpdfbookmark{Theorem~\ref{thm:main}}{thm:main}
\begin{maintheorem}
    \label{thm:main}
    Let \(\zeta\supbulk\in \ell^1\) and \((H_L(b))_{L\in\N}\) be a family of Hamiltonians of the form~\eqref{eq:Hb}.
    Then, there exists a null sequence \(\theta\) and a constant \(C>0\) such that the following holds:
    If \((H_L(b))_{L\in \N}\) satisfies local indistinguishability of the Gibbs state in the bulk at \((\beta,\mu,b)\) with \(\zeta\supbulk\)-decay in the sense of Definition~\ref{def:local-indistinguishability}, then
    \begin{equation}\label{eq:thm-m_L-equals-I_L}
        \bigl\lvert
            m_L(\beta,\mu,b) - \Iedge[L](\beta,\mu,b)
        \bigr\rvert
        \leq
        \theta(L)
        \qquad\text{for all \(L \geq D+R\).}
    \end{equation}
    Moreover, the edge current is localized near the edge in the sense that for all \(L \geq d \geq R+D\)
    \begin{equation}\label{eq:thm-localization-I_L}
        \bigl\lvert
            \Iedged[L]{d}(\beta,\mu,b) - \Iedge[L](\beta,\mu,b)
        \bigr\rvert
        \leq
        C \smashoperator[lr]{\sum_{n=d-R-D}^\infty} \zeta\supbulk(n).
    \end{equation}
\end{maintheorem}

\begin{remark}
    \label{rem:decay-function-theta}
    The sequence \(\theta\) is given as a function of \(\zeta\supbulk\) by~\eqref{eq:theta-def}.
    If \(\zeta\supbulk(r) \leq C' \, (r+1)^{-(n+1)}\) with \(n>1\), then \(\theta(L) \leq C'' \, L^{-n/(n+2)}\) and~\eqref{eq:thm-localization-I_L} is bounded by \(C''' \, (d-R-D)^{-n}\).
    While~\eqref{eq:thm-localization-I_L} scales basically like~\(\zeta\supbulk\), the best possible decay in~\eqref{eq:thm-m_L-equals-I_L} is~\(\theta(L) \sim 1/L\), which results from the fact that the fraction of the total area occupied by the boundary scales like~\(1/L\) in two dimensions.
    However, unless both the magnetization and the edge current vanish, no better decay can be expected even in the non-interacting case, cf.~\cite{CMT2021}.
\end{remark}

Then, if we further assume local indistinguishability of the Gibbs state everywhere we can also analyze the thermodynamic limit of~\eqref{eq:thm-localization-I_L}.
First, notice that if the family of Hamiltonians \((H_L(b))_{L\in\N}\) satisfies local indistinguishability with \(\zeta\supedge\)-decay everywhere, then for all finite \(X\subset \Z^2_+\) and all observables \(A\in \alg_X\) the expectation values \(\tr\bigl(\rho_L(\beta,\mu,b)\, A\bigr)\), which are defined if \(\Lambda_L \supset X\), form a Cauchy sequence in \(L\), and we define \(\rho_\infty(\beta,\mu,b)(A)\) to be its limit.
Thus, there exists a unique thermodynamic limit state \(\rho_\infty(\beta,\mu,b)\) on \(\alg_\mathrm{loc}\), and we can define the edge current in this state by
\begin{equation}
    \Iedged{d}(\beta,\mu,b)
    :=
    \rho_\infty(\beta,\mu,b)\biggl(
        \sum_{n=0}^{d} J_{1}^{(0,n)}(b)
    \biggr).
\end{equation}

\belowpdfbookmark{Theorem~\ref{thm:limit-I}}{thm:limit-I}
\begin{maintheorem}
    \label{thm:limit-I}
    Let \(\zeta\supbulk\in \ell^1\) and \(\zeta\supedge\) tend to zero.
    Let \((H_L(b))_{L\in\N}\) be a family of Hamiltonians of the form~\eqref{eq:Hb} satisfying local indistinguishability of the Gibbs state at \((\beta,\mu,b)\) with \(\zeta\supbulk\)-decay in the bulk and \(\zeta\supedge\)-decay everywhere, in the sense of Definition~\ref{def:local-indistinguishability}.
    Then the thermodynamic limit
    \begin{equation}
        \label{eq:thm-limit-magnetization}
        m(\beta,\mu,b) := \lim_{L\to\infty} m_L(\beta,\mu,b)
    \end{equation}
    exists,
    and the total edge current
    \begin{equation} \label{eq:definition-infinite-edge-current}
        \Iedge(\beta,\mu,b)
        :=
        \lim_{d\to\infty} \, \Iedged[]{d}(\beta,\mu,b),
    \end{equation}
    satisfies
    \begin{equation*}
        m(\beta,\mu,b) = \Iedge(\beta,\mu,b).
    \end{equation*}
    The edge current is localized near the edge in the sense that there is \(C>0\) so that for all \(d \geq D + R\)
    \begin{equation*}
        \bigl\lvert
            \Iedged{d}(\beta,\mu,b) - \Iedge(\beta,\mu,b)
        \bigr\rvert
        \leq
        C \smashoperator[lr]{\sum_{n=d-R-D}^\infty} \zeta\supbulk(n).
    \end{equation*}
    Moreover, \(m(\beta,\mu,b)\) and \(\Iedge(\beta,\mu,b)\) are independent of the specific edge contributions~\(T_b\supedge\) and~\(\Phi\supedge\).
\end{maintheorem}
The precise statement on the independence from the edge Hamiltonian is given in Proposition~\ref{prop:different-edges}, where we also derive an explicit bound on the difference of the finite-volume edge currents, see~\eqref{eq:boundary_independence}.

Note, that local indistinguishability everywhere with \(\zeta\supedge\)-decay implies local indistinguishability in the bulk with decay at most \(\zeta\supbulk \leq \zeta\supedge\).
But we assume the decay in the bulk separately in order to take into account the possibility that the localization properties of the edge current might be better than its speed of convergence in the thermodynamic limit, cf.\ Proposition~\ref{prop:I_infinity-Bloch-theorem-localization-convergence}, which is the situation in the non-interacting setting~\cite{CMT2021,MS2022}.

\begin{remark}\label{rem:theorem-bulk-edge-correspondence}
    We present Theorems~\ref{thm:main} and~\ref{thm:limit-I} with an edge interaction only at the lower boundary for simplicity.
    The proofs allow for a more general setting where an interaction can be added on all four sides.
    Moreover, then the magnetization \(m(\beta,\mu,b)\) and the edge current \(\Iedge[](\beta,\mu,b)\) are independent of these boundary perturbations.
    More precisely, we have the following:
    For each \(L\in \N\) let \(\Phi\supboundary_L\) be a finite-range interaction supported on \(Q_L := \Lambda_L \setminus [-L+D,L-D]\times[0,2L-D]\), a strip of width \(D\) around the remaining boundaries of \(\Lambda_L\), commuting with all local number operators \(\mathcal{N}_{\{z\}}\) and with interaction terms bounded uniformly in \(L\), and let \(T_b\supboundary\) be a hopping as in~\eqref{eq:definition-Tb} with \(T\supboundary\) supported on \(Q_L\).
    Let \((H_L(b))_{L\in\N}\) be a family of Hamiltonians of the form~\eqref{eq:Hb} and let
    \begin{equation*}
        \tilde H_L(b)
        :=
        H_L(b) + \sum_{x,y \in \Lambda_L} a^*_x \, T\supboundary_b(x, y) \, a^{}_y + \sum_{X\subset \Lambda_L} \Phi\supboundary_L(X)
        .
    \end{equation*}
    Then the respective statements of Theorem~\ref{thm:main} and~\ref{thm:limit-I} also hold for the family of Hamiltonians \((\tilde H_L(b))_{L\in\N}\).
    Moreover, if both \(\tilde H_L(b)\) and \(H_L(b)\) satisfy local indistinguishability, then \(\tilde m(\beta,\mu,b) = m(\beta,\mu,b)\) and \(\tIedge(\beta,\mu,b) = \Iedge(\beta,\mu,b)\) are independent of the boundary terms.
\end{remark}

Supposing that the assumptions of Theorem~\ref{thm:limit-I} hold uniformly in some open interval of chemical potentials around \(\mu\), we obtain differentiability of \(m(\beta,\mu,b)\) with respect to~\(\mu\) and thus, as explained in the introduction, a further step in the direction of proving the equality of transport coefficients.

\belowpdfbookmark{Theorem~\ref{thm:muderivative}}{thm:muderivative}
\begin{maintheorem}
    \label{thm:muderivative}
    Let \(n\mapsto n^2\,\zeta\supedge(n)\in \ell^1\).
    Let \((H_L(b))_{L\in\N}\) be a family of Hamiltonians of the form~\eqref{eq:Hb} satisfying local indistinguishability of the Gibbs state everywhere with \(\zeta\supedge\)-decay at \((\beta,\mu,b)\), for \(\mu\) in some open interval.
    Then \(m(\beta,\mu,b)\) and \(\Iedge(\beta,\mu,b)\) defined in Theorem~\ref{thm:limit-I} are differentiable, and thus
    \begin{equation*}
        \dmu\, m(\beta,\mu,b)
        =
        \dmu\, \Iedge(\beta,\mu,b)
        .
    \end{equation*}
\end{maintheorem}

\begin{remark}
    As we already know from Theorem~\ref{thm:limit-I} that \(m(\beta,\mu,b)=\Iedge(\beta,\mu,b)\), Theorem~\ref{thm:muderivative} will follow from differentiability of \(\Iedge(\beta,\mu,b)\), see Proposition~\ref{prop:differentiability_I}.
    Additionally, we prove a quantitative bound for the difference of the two quantities in finite volume in Propositions~\ref{prop:dmu-m_L-equals-dmu-I_L}, and localization near the boundary in Proposition~\ref{prop:differentiability_I}.
\end{remark}

An important ingredient in the proof of Theorem~\ref{thm:main} is the vanishing of the equilibrium current in the bulk.
This result is known in the literature as Bloch’s theorem (see e.g.~\cite{Watanabe2019,BF2021} and references therein) and has an importance on its own.
In our setting with open boundary conditions, it is just a consequence of current conservation coupled with the local indistinguishability of the Gibbs state.
This allows for better decay rates than in the setting with periodic boundary conditions.

\begin{restatable}[Bloch’s Theorem]{proposition}{blochtheorem}
    \label{prop:Bloch-theorem}
    Let \(\zeta\supbulk\) be a null sequence and \((H_L(b))_{L\in\N}\) be a family of Hamiltonian of the form~\eqref{eq:Hb}.
    There exists \(\Cb>0\) such that the following holds:
    If \((H_L(b))_{L\in\N}\) satisfies local indistinguishability of the Gibbs state in the bulk at \((\beta,\mu,b)\) with \(\zeta\supbulk\)-decay in the sense of Definition~\ref{def:local-indistinguishability}, then
    \begin{equation}\label{eq:prop-Bloch-theorem}
        \bigl\lvert \tr \bigl( \rho_L(\beta,\mu,b) \, J_{k,L}^{z}(b) \bigr) \bigr\rvert
        \leq
        \Cb \, \zeta\supbulk\Bigl(\bigl[\dist(z,\Z^2\setminus\Lambda_L)-D-R\bigr]_+\Bigr).
    \end{equation}
\end{restatable}

We have established the equality of the edge current and the magnetization and proven that the edge current is an edge quantity in the sense, that it is localized near the edge.
Finally, we argue that the magnetization is a bulk quantity by showing that it can be obtained directly in the infinite volume without any edge.
Denote by \(B_L := [-L,L]^2 \cap \Z^2\) the boxes centered on the origin and let
\begin{equation}
    \label{eq:Hbulk}
    H\supbulk_L(b)
    =
    \sumstack[lr]{x,y \in B_L} a^*_x \, T\supbulk_b(x, y) \, a^{}_y
    + \sumstack[lr]{X\subset B_L} \Phi\supbulk(X)
\end{equation}
be the bulk Hamiltonian and
\begin{equation}
    \label{eq:dynamics-Hbulk}
    \tau_t\supbulk(A)
    =
    \lim_{L\to \infty}
    \E^{\I t (H\supbulk_L(b) - \mu \mathcal{N}_{B_L})}
    \, A
    \, \E^{-\I t (H\supbulk_L(b)- \mu \mathcal{N}_{B_L})}
\end{equation}
be the infinite-volume dynamics generated by the bulk Hamiltonian (adjusted by the chemical potential).
Following the arguments from~\cite{AM2003} we will first show that the pressure of the infinite volume limits of the edge and bulk system agree.
Under a somewhat stronger hypothesis and using ideas similar to Theorem~\ref{thm:muderivative}, the magnetization \(m(\beta,\mu,b)\) can then also be obtained directly in the infinite-volume system.

To this end, note that the pressure for general states \(\rho\) on \(\alg_{B_L}\) is defined as
\begin{equation}\label{eq:definition-pressure-state}
    P(\rho)
    :=
    \tr\Bigl(\rho \, \bigl(H\supbulk_L(b) - \mu \, \mathcal{N}_{B_L}\bigr)\Bigr) - \beta^{-1} \, S(\rho)
    \quadtext{with}
    S(\rho) := -\tr( \rho \ln \rho )
    ,
\end{equation}
which agrees with \(-\beta^{-1} \, \ln \mathcal{Z}_{B_L}\bigl[H\supbulk_L(b)\bigr](\beta,\mu)\) for the Gibbs state of \(H\supbulk_L(b)\), compare~\eqref{eq:definition-pressure-finite-volume-partition-function}.
The following theorem states that the pressure per unit volume, \(\lim_{L\to \infty}P(\omega|_{B_L})/(2L+1)^2\), where \(\omega|_{B_L}\) is the restriction of a bulk equilibrium state \(\omega\) to \(\alg_{B_L}\), equals the thermodynamic limit of the pressure in the system with an edge.

\belowpdfbookmark{Theorem~\ref{thm:limit-p}}{thm:limit-p}
\begin{maintheorem}
    \label{thm:limit-p}
    Let \((H_L(b))_{L\in\N}\) be a family of Hamiltonians of the form~\eqref{eq:Hb}.
    For any \(\beta>0\), \(\mu\), \(b\in \R\) the thermodynamic limit
    \begin{equation*}
        p(\beta,\mu,b) := \lim_{L\to\infty} p_L(\beta,\mu,b)
    \end{equation*}
    of the pressure exists and is independent of the boundary terms.
    Moreover, for any \((\tau\supbulk,\beta)\)-KMS state \(\omega\) the pressure per volume of \(\omega\)
    equals \(p(\beta,\mu,b)\),
    \begin{equation*}
        p(\beta,\mu,b)
        =
        \lim_{L\to \infty} \frac{P\bigl(\omega|_{B_L}\bigr)}{(2L+1)^2}.
    \end{equation*}

    Additionally, given \(\zeta\supbulk\in \ell^1\) and \(\zeta\supedge\) tending to zero, assume that
    \(\bigl(H_L(b)\bigr)_{L\in \N}\) satisfies local indistinguishability in the bulk with \(\zeta\supbulk\)-decay and everywhere with \(\zeta\supedge\)-decay at \((\beta,\mu,b)\), for \(b\) in some open interval.
    Then \(b\mapsto p(\beta,\mu,b)\) is differentiable and its derivative agrees with the magnetization \(m(\beta,\mu,b)\) defined in~\eqref{eq:thm-limit-magnetization},
    \begin{equation*}
        \partial_b \, p(\beta,\mu,b) = m(\beta,\mu,b)
        .
    \end{equation*}
\end{maintheorem}

\section{Proofs}

\subsection{Bloch's theorem}

This section is devoted to the proof of Bloch's theorem, namely Proposition~\ref{prop:Bloch-theorem}.
The proof is based on the bulk homogeneity of the system and on the continuity equation for the current provided by Lemma~\ref{lem:continuity-equation}, together with the local indistinguishability assumption for the Gibbs state.
The homogeneity in the bulk of the system is encoded in the invariance under magnetic translation, which we briefly recall in the next section.

\subsubsection{Magnetic translations}
\label{sec:magnetic-translations}

On the one-particle Hilbert space \(\mathfrak{h}=\ell^2(\Z^2)\) of the full lattice \(\Z^2\), for \(y\in \Z^2\) and \(b>0\) the magnetic translation \(U_y(b)\) is defined by its action on \(\psi \in \mathfrak{h}\) as
\begin{equation*}
    \bigl( U_y(b) \, \psi \bigr)(x)
    =
    \E^{\I b y_2 x_1} \, \psi(x-y)
    \quadtext{with adjoint}
    \bigl( U_y^*(b) \, \psi \bigr)(x)
    =
    \E^{-\I b y_2 (x_1+y_1)} \, \psi(x+y)
    .
\end{equation*}
Then, in second quantization, for which we use the same symbol, we obtain
\begin{equation*}
    U_y^*(b) \, a_x^* \, U_y^{}(b)
    =
    \E^{-\I y_2 b x_1} \, a_{x-y}^*.
\end{equation*}

With this definition, the kinetic part of the bulk Hamiltonian~\eqref{eq:Hb} is invariant under magnetic translations since
\begin{align*}
    U_z^*(b) \, a^*_x \, T\supbulk_b(x,y) \, a^{}_y \, U_z^{}(b)
    &=
    U_z^*(b) \, a_x^* \, U_z^{}(b) \, \E^{ \I b \frac{x_2+y_2 \xmathstrut{0.2}}{2} (x_1-y_1)} \, T\supbulk(x,y) \, U_z^*(b)\, a^{}_y \, U_z^{}(b)
    \\&=
    a^*_{x-z} \, \E^{ \I b \frac{x_2+y_2 \xmathstrut{0.2}}{2} \, (x_1-y_1) - \I b \frac{2z_2x_1 \xmathstrut{0.2}}{2} + \I b \frac{2z_2y_1 \xmathstrut{0.2}}{2}} T\supbulk(x-z,y-z) \, a^{}_{y-z}
    \\&=
    a^*_{x-z} \, T\supbulk_b(x-z,y-z) \, a^{}_{y-z}.
\end{align*}
Moreover, we assume that the bulk interaction \(\Phi\supbulk\) is invariant under magnetic translations, namely
\begin{equation}
    \label{eq:invariance-under-magnetic-translations}
    U_z^*(b) \, \Phi\supbulk(X) \, U_z(b) = \Phi\supbulk(X-z).
\end{equation}
Hence, the complete bulk part of the Hamiltonian~\eqref{eq:Hb} is invariant under magnetic translations, and
\begin{equation*}
    U_z^*(b) \, H_L(b)\big|_{X} \, U_z(b) = H_L(b)\big|_{X-z}
\end{equation*}
for \(X\subset \Lambda_L\) and \(z\in \Z^2\) such that \(\dist(X,\Z^2\setminus\Lambda_L)\) and \(\dist(X-z,\Z^2\setminus\Lambda_L)>D\).
This property carries over to the local Gibbs state in the sense, that
\begin{equation*}
    U_z^*(b) \, \rho\bigl[H_L(b)\big|_{X}\bigr](\beta,\mu) \, U_z(b)
    = \rho\bigl[H_L(b)\big|_{X-z}\bigr](\beta,\mu).
\end{equation*}

Moreover, the above calculation also shows that \(U_{z}^*(b) \, J_{k}^{y}(b) \, U_{z}(b) = J_{k}^{y-z}(b)\) for all \(y_2,z_2 \geq R+D\), due to the simple structure of \(J_{k}^{y}(b)\).
Then, it follows that
\begin{equation}
    \label{eq:translation-invariance-local-current-expectation}
    \tr \bigl(\rho\bigl[H_L(b)\big|_{X}\bigr](\beta,\mu) \, J_{1}^{y}(b)\bigr)
    =
    \tr \bigl(
        \rho\bigl[H_L(b)\big|_{X-z}\bigr](\beta,\mu)
        \, J_{1}^{y-z}(b)
    \bigr)
\end{equation}
if \(X\), \(y\) and \(z\) fulfill all mentioned conditions.

\subsubsection{The continuity equation}
\label{sec:continuity-equation}

We prove the continuity equation for the local currents defined in~\eqref{eq:current-through-dual} and the resulting current conservation for stationary states.
These two facts play a key role in the proof of Bloch's theorem.

For this purpose, let us define the dual edge boundary of a set \(Z\subset \Lambda_L\) as
\begin{equation*}
    \partial_{\Lambda_L} Z
    :=
    \biggl\{\,
        (k,z) \in \{1,2\} \times \biggl(
            \begin{gathered}
                \{-L,\dotsc,L-1\}
                \\
                \LTadd{&}\times\{0,\dotsc,2L-1\}
            \end{gathered}
        \biggr)
    \,\bigg|\,
        \begin{gathered}
            z \in Z \text{ and } z+\unitvec{k} \in \Lambda_L \setminus Z
            \quad\text{or}
            \\
            z \in \Lambda_L \setminus Z \text{ and } z+\unitvec{k} \in Z\phantom{\quad\text{or}}
        \end{gathered}
    \,\biggr\}.
\end{equation*}
This is exactly the set of labels \((k,z)\), such that the union of the dual edges \(\overline{\edge_{k,z}}\) is the boundary of the set \(\bigcup_{z\in Z} z+ [-1/2,1/2]^2\) in \(\bigcup_{z\in \Lambda_L} z+ [-1/2,1/2]^2\).

\begin{lemma}[Continuity equation]
    \label{lem:continuity-equation}
    For any \(z\in \Lambda_L\), the currents defined in~\eqref{eq:current-through-dual} satisfy the continuity equation
    \begin{equation} \label{eq:lem-continuity-equation}
        \begin{aligned}
            \frac{\D}{\D t} \,
            \E^{\I H_L(b) \, t} \, \mathcal{N}_{\{z\}} \, \E^{-\I H_L(b) \, t}
            \Bigr|_{t=0}
            & =
            \operatorname{div}_z J_{k,L}^z(b)
            \\& :=
            J_{1,L}^z(b) - J_{1,L}^{z-\unitvec{1}}(b)
            + J_{2,L}^z(b) - J_{2,L}^{z-\unitvec{2}}(b) .
        \end{aligned}
    \end{equation}
\end{lemma}

\begin{proof}
    Calculating the derivative, we find
    \begin{align*}
        \frac{\D}{\D t} \,
        \E^{\I H_L(b) \, t} \, \mathcal{N}_{\{z\}} \, \E^{-\I H_L(b) \, t}
        \Bigr|_{t=0}
        &=
        \I \, [H_L(b), \mathcal{N}_{\{z\}}]
        \\&=
        \I \sumstack[lr]{x,y\in \Lambda_L}
        \, \bigl[ a^*_x \, T_b(x, y) \, a^{}_y, \mathcal{N}_{\{z\}} \bigr]
        \\&=
        \I \sumstack[lr]{x\in \Lambda_L\setminus\{z\}}
        \, a^*_x \, T_b(x, z) \, a^{}_z
        - \I \sumstack[lr]{y\in \Lambda_L\setminus\{z\}}
        a^*_z \, T_b(z, y) \, a^{}_y
        \\&=
        \I \sumstack[lr]{x\in \Lambda_L\setminus\{z\}}
        a^*_x \, T_b(x, z) \, a^{}_z
        - a^*_z \, T_b(z, x) \, a^{}_x
        ,
    \end{align*}
    where we used that \(\bigl[\sum_{Z\subset \Lambda_L}\Phi\supbulk(Z), \mathcal{N}_{\{z\}}\bigr]=\bigl[\sum_{Z\subset \Lambda_L}\Phi\supedge(Z), \mathcal{N}_{\{z\}}\bigr]=0\) by assumption.

    It is left to rewrite the sum in terms of \(J_{k,L}^z(b)\).
    Each \(\overline{xz}\) in the sum will cross the rectangle around \(z\) formed by the four dual edges \(\overline{\edge_{1,z}}\), \(\overline{\edge_{1,z-\unitvec{1}}}\), \(\overline{\edge_{2,z}}\) and \(\overline{\edge_{2,z-\unitvec{2}}}\) of the dual lattice at one point.
    If this point lies within a dual edge \(\edge_{k,q}\), the term in the sum will contribute to \(J_{k,L}^q(b)\) with weight \(1\).
    Otherwise, \(\overline{xz}\) crosses the rectangle at a corner and the contribution is attributed evenly to the two adjacent dual edges with weight \(1/2\).
    The sum on the right-hand side of~\eqref{eq:lem-continuity-equation} still contains some more terms coming from elements \(x,y\in \Lambda_L\setminus\{q\}\) such that the line \(\overline{xy}\) intersects two of the four dual edges.
    As can be easily checked, the corresponding contributions come with different signs and cancel each other.
    For example, the term \(a_x^* \, T_b(x,y) \, a_y^{}\) for \(x=z+\unitvec{2}\) and \(y=z+\unitvec{1}\) appears with negative sign in \(J_{1,L}^{z}(b)\), because \(x_1-y_1<0\), and with positive sign in \(J_{2,L}^{z}(b)\), because \(x_2-y_2>0\).
\end{proof}

As a simple consequence, in a stationary state it follows that the net current into any set \(Z\subset \Lambda_L\) is zero.
This is an important ingredient for the following proof.

\begin{corollary}[Current conservation]
    \label{cor:current-conservation}
    For any \(Z\subset \Lambda_L\) and stationary state%
    \footnote{ 
        The statement actually holds for all bounded operators \(\rho\in \alg_{\Lambda_L}\) but is naturally interesting for states or similar objects (see section~\ref{sec:proof-theorem-muderivative}).
    }
    \(\rho\), current conservation holds
    \begin{equation}
        \label{eq:lem-current-conservation}
        \sum_{(k,z)\in \partial_{\Lambda_L} Z}
        (-1)^{\delta_{z\in Z}} \, \tr \bigl(\rho \, J_{k,L}^z(b)\bigr)
        =
        0.
    \end{equation}
    Here \(\delta_{z\in Z} = 1\) if \(z\in Z\) and \(0\) otherwise, takes the role of the normal vector in the continuous analogue.
\end{corollary}

\begin{proof}
    Taking the expectation value of~\eqref{eq:lem-continuity-equation} and summing over \(z\in Z\) yields
    \begin{equation*}
        0
        =
        \frac{\D}{\D t} \,
        \tr \Bigl(
            \E^{-\I H_L(b) \, t} \, \rho \, \E^{\I H_L(b) \, t}
            \, \mathcal{N}_{Z}
        \Bigr)
        \Bigr|_{t=0}
        =
        \sumstack{k\in \{1,2\}\\ z\in Z}
        \tr\bigl(\rho \, J_{k,L}^z(b)\bigr)
        - \tr\bigl(\rho \, J_{k,L}^{z-\unitvec{k}}(b)\bigr)
        ,
    \end{equation*}
    due to stationarity of \(\rho\) and cyclicity of the trace.
    In the sum, the positive term for \(z\in Z\) is cancelled by the negative one for \(z+\unitvec{k}\in Z\) and only~\eqref{eq:lem-current-conservation} remains.
\end{proof}

\subsubsection{Proof of Proposition~\ref{prop:Bloch-theorem}}
\label{sec:proof-Bloch-theorem}

\blochtheorem*

\begin{proof}
    We do the proof for \(k=1\), i.e.\ currents in \(x_1\)-direction, since the case \(k=2\) is analogous.
    Let \(d > D + R\).
    By current conservation for the rectangle \(\Lambda_L \cap \{x_1 \geq 0\}\), whose boundary in \(\Lambda_L\) is simply the vertical line at \(x_1=0\) (see edge set~(a) in Figure~\ref{fig:current-conservation} for \(m=0\)), we find
    \begin{align*}
        0
        &= \sum_{n=0}^{2L} \tr \bigl(\rho_L(\beta,\mu,b) \, J_{1,L}^{(0,n)}(b)\bigr)
        \\&=
        \begin{aligned}[t]
            & \sum_{n=0}^{d-1}
            \tr \bigl(\rho_L(\beta,\mu,b) \, J_{1,L}^{(0,n)}(b)\bigr)
            + \tr \bigl(\rho_L(\beta,\mu,b) \, J_{1,L}^{(0,2L-n)}(b)\bigr)
            \\&+ \sum_{n=d}^{2L-d}
            \tr \bigl(\rho_L(\beta,\mu,b) \, J_{1}^{(0,n)}(b)\bigr)
            - \tr \bigl(\rho\bigl[H_L(b)\big|_{B^{(0,n)}(\ell)}\bigr](\beta,\mu) \, J_{1}^{(0,n)}(b)\bigr)
            \\&+ \sum_{n=d}^{2L-d}
            \tr \bigl(\rho\bigl[H_L(b)\big|_{B^{(0,n)}(\ell)}\bigr](\beta,\mu) \, J_{1}^{(0,n)}(b)\bigr)
        \end{aligned}
        \\&=:
        A_1 + A_ 2 + A_3,
    \end{align*}
    where \(B^x(\ell):= \{\, y\in \Z^2 \,\vert\, \dist(x,y) \leq \ell \,\}\) is the ball in \(\Z^2\), \(R<\ell<d-D\) and we replaced \(J^x_{1,L}(b)=J^x_{1}\) in the last two sums in view of the previous remark that the \(L\)-dependence of \(J_{k,L}^z(b)\) only stems from missing hopping terms near the boundary.
    Due to their structure, \(J_{1,L}^{(m,n)}(b)\) are bounded operators with norm bound \(\lVert J_{1,L}^{(m,n)}(b) \rVert < C_J\) uniform for all \(m\) and \(n\).
    Hence, also \(\bigl\lvert \tr \bigl(\rho_L(\beta,\mu,b) \, J_{1,L}^{(m,n)}(b)\bigr) \bigr\rvert < C_J\) and \(\lvert A_1 \rvert\) is bounded by \(2 \, d \, C_J\).
    For the second sum \(A_2\) notice that \(J_{1,L}^{(m,n)}(b)\in \alg_{B^{(m,n)}(R)}\).
    Thus, by using the local indistinguishability of the Gibbs state, see Definition~\ref{def:local-indistinguishability}, we have
    \begin{equation*}
        \lvert A_2 \rvert
        \leq
        \sum_{n=d}^{2L-d} C_J \, \zeta\supbulk(\ell-R)
        =
        (2(L-d)+1) \, C_J \, \zeta\supbulk(\ell-R)
        .
    \end{equation*}
    And by the translation invariance of the Gibbs state, namely~\eqref{eq:translation-invariance-local-current-expectation}, we get
    \begin{equation*}
        (2\,(L-d)+1) \, \tr \Bigl(\rho\bigl[H_L(b)\big|_{B^{(0,L)}(\ell)}\bigr](\beta,\mu) \, J_{1}^{(0,L)}(b)\Bigr)
        =
        A_3
        =
        -(A_1+A_2)
        ,
    \end{equation*}
    which together with the previous bound on \(A_1\) and \(A_2\), implies
    \begin{equation*}
        \Bigl\lvert \tr \Bigl(\rho\bigl[H_L(b)\big|_{B^{(0,L)}(\ell)}\bigr](\beta,\mu) \, J_{1}^{(0,L)}(b)\Bigr) \Bigr\rvert
        \leq
        \frac{C_J \, d}{L-d} + C_J \, \zeta\supbulk(\ell-R)
        .
    \end{equation*}
    As the left-hand side is actually independent of \(L\), it is bounded by the infimum \(C_J \, \zeta\supbulk(\ell-R)\) of the right-hand side.

    We can now prove the same for every \(J_{1,L}^{(m,n)}(b)\) with \(z=(m,n)\) in a finite box by using~\eqref{eq:translation-invariance-local-current-expectation} together with
    \begin{equation*}
        \rho\bigl[H_L(b)\big|_{B^{(m,n)}(\ell)}\bigr](\beta,\mu) \, J_{1}^{(m,n)}(b)
        = \rho\bigl[H_{L'}(b)\big|_{B^{(m,n)}(\ell)}\bigr](\beta,\mu) \, J_{1}^{(m,n)}(b)
    \end{equation*}
    if \(\dist\bigl((m,n),\Z^2\setminus\Lambda_L\bigr)>R+D\).
    Indeed, we have
    \begin{align*}
        &\alignindent
        \Bigl\lvert
            \tr \bigl(\rho_L(\beta,\mu,b) \, J_{1,L}^{(m,n)}(b)\bigr)
        \Bigr\rvert
        \\&\leq
        \begin{aligned}[t]
            &\Bigl\lvert
                \tr \bigl(\rho_L(\beta,\mu,b) \, J_{1}^{(m,n)}(b)\bigr)
                - \tr \bigl(\rho\bigl[H_L(b)\big|_{B^{(m,n)}(\ell)}\bigr](\beta,\mu) \, J_{1}^{(m,n)}(b)\bigr)
            \Bigr\rvert
            \\&+
            \Bigl\lvert
                \tr \bigl(\rho\bigl[H_{L}(b)\big|_{B^{(0,L)}(\ell)}\bigr](\beta,\mu) \, J_{1}^{(0,L)}(b)\bigr)
            \Bigr\rvert
        \end{aligned}
        \\&\leq
        2 \, C_J \, \zeta\supbulk(\ell-R)
    \end{align*}
    by using local indistinguishability and the bound for \(J_{1,L}^{(0,L)}(b)\).
    We can now choose \(\ell = \dist\bigl((m,n),\Z^2\setminus\Lambda_L\bigr)-D\), which proves the statements for \(\dist\bigl((m,n),\Z^2\setminus\Lambda_L\bigr) > R + D\).
    And since we argued above that the LHS of~\eqref{eq:prop-Bloch-theorem} is in any case bounded for all \(z\in\Lambda_L\), the full statement follows with \(\Cb = C_J \, \max\{\zeta\supbulk(0)^{-1},2\}\).
\end{proof}

\subsection{Proof of Theorem~\ref{thm:main}}

We split the proof of Theorem~\ref{thm:main} in two parts: we briefly discuss the localization of the edge current first and then the equality between the magnetization and the current.

\subsubsection{Localization of the current}

The localization of the edge current near the edge is a straightforward consequence of \nameref{prop:Bloch-theorem}.

\begin{proposition}
    \label{prop:localization-I_L}
    Let \(\zeta\supbulk\in \ell^1\), \((H_L(b))_{L\in\N}\) be a family of Hamiltonians of the form~\eqref{eq:Hb} and \(\Cb\) be the constant from \nameref{prop:Bloch-theorem}.
    If \((H_L(b))_{L\in\N}\) satisfies local indistinguishability of the Gibbs state in the bulk at \((\beta,\mu,b)\) with \(\zeta\supbulk\)-decay in the sense of Definition~\ref{def:local-indistinguishability}, then for all \(L \geq d \geq D+R\),
    \begin{equation*}
        \bigl\lvert
            \Iedged[L]{d}(\beta,\mu,b) - \Iedge[L](\beta,\mu,b)
        \bigr\rvert
        \leq
        \Cb \smashoperator[lr]{\sum_{n=d-D-R}^\infty} \zeta\supbulk(n).
    \end{equation*}
\end{proposition}

\begin{proof}
    For the proof we just apply \nameref{prop:Bloch-theorem} to obtain
    \begin{align*}
        \bigl\lvert
            \Iedged[L]{d} - \Iedge[L]
        \bigr\rvert
        \leq
        \sum_{n=d}^{L-1}
        \, \Bigl\lvert
            \tr\Bigl( \rho_L(\beta,\mu,b) \, J_{1}^{(0,n)}(b) \Bigr)
        \Bigr\rvert
        \leq
        \Cb \, \sum_{n=d}^{\infty}
        \, \zeta\supbulk(n-D-R)
        .
    \end{align*}
\end{proof}

\subsubsection{Magnetization in finite systems}
\label{sec:proof-theorem-m_L-equals-I_L}

Let us compute the magnetic derivative of the Hamiltonian.
Notice that for every fixed \(L\) the Hamiltonian \(H_L(b)\) is a smooth function of \(b\) in the operator norm topology.
Taking into account~\eqref{eq:definition-Tb}, we find that
\begin{align}
    H'_L(b)
    := \frac{\partial}{\partial b} H_L(b)
    &= \sum_{x \in \Lambda_L} \sum_{y \in B_L^x(R)} \tfrac{\I}{2} \, (x_2+y_2) \, (x_1-y_1) \, a_x^*\, T_b (x,y) \,a^{}_y
    \nonumber
    \\&= \sum_{m=-L}^{L-1} \sum_{n=0}^{2L} n \, J_{1,L}^{(m,n)}(b).
    \label{eq:magnetic-derivative-H}
\end{align}
Where \(J_{1,L}^{(m,n)}(b)\) has been defined in~\eqref{eq:current-through-dual}.
To see the last equality, we compare the coefficients in front of \(a_x^*\, T_b (x,y) \,a^{}_y\).
Without loss of generality we only consider \(x_1 < y_1\) and \(x_2 \leq y_2\).
By~\eqref{eq:current-through-dual}, the coefficient on the right-hand side is
\begin{equation*}
    - \I \sum_{m=x_1}^{y_1-1} \tfrac{1}{2} \, \biggl(
        \sumstack{n\in \N:\\\overline{xy}\,\cap\, \edge_{1,(m,n)} \neq \emptyset} \mathllap{n}
        \quad+\quad
        \sumstack[l]{n\in \N:\\\overline{xy}\,\cap\, \overline{\edge_{1,(m,n)}} \neq \emptyset} \mathllap{n}
    \biggr)
    .
\end{equation*}
By point symmetry around the center \((x+y)/2\) of \(\overline{xy}\), whenever \(\overline{xy}\,\cap\, \overline{\edge_{1,(m,n)}} \neq \emptyset\) for \(m = x_1 + k\) and \(n\in \N\), then also \(\overline{xy}\,\cap\, \overline{\edge_{1,(m',n')}} \neq \emptyset\) for \(m' = y_1 - k - 1\) and \(n' = y_2 - (n - x_2)\).
The same holds for the edges without the endpoints.
Thus, the coefficient equals
\begin{equation*}
    - \I \, \tfrac{1}{2} \, \sum_{k=0}^{\crampedclap{y_1-x_1-1}} \tfrac{1}{2} \, \biggl(
        \sumstack[r]{n\in \N:\\\overline{xy}\,\cap\, \edge_{1,(x_1+k,n)}\neq \emptyset} n + \bigl(y_2-(n-x_2)\bigr)
        +
        \sumstack[lr]{n\in \N:\\\overline{xy}\,\cap\, \overline{\edge_{1,(x_1+k,n)}} \neq \emptyset} n + \bigl(y_2-(n-x_2)\bigr)
    \biggr)
    =
    \tfrac{\I}{2} \, (x_1-y_1) \, (y_2 + x_2)
    .
\end{equation*}
This is exactly the coefficient on the left-hand side of~\eqref{eq:magnetic-derivative-H}.

By using Duhamel’s formula and~\eqref{eq:magnetic-derivative-H}, we can explicitly compute the magnetization as follows
\begin{align}
    m_L(\beta,\mu,b)
    &= \nonumber
    - (2L+1)^{-2} \, \beta^{-1} \, \frac{\partial}{\partial b}
    \ln \Bigl( \tr\bigl( \E^{-\beta (H_L(b) - \mu \,\mathcal{N}_L)} \bigr)\Bigr)
    \\ &= \nonumber
    -\frac{1}{(2L+1)^2 \, \beta \, \mathcal{Z}_L(\beta,\mu,b)}
    \tr\biggl( \frac{\partial}{\partial b} \, \E^{-\beta (H_L(b) - \mu\, \mathcal{N}_L)} \biggr)
    \\ &= \nonumber
    \frac{1}{(2L+1)^2 \, \mathcal{Z}_L(\beta,\mu,b)}
    \tr\biggl( \int_0^1 \D s \, \E^{-s\beta (H_L(b) - \mu \,\mathcal{N}_L)} \, H'_L(b) \, \E^{-(1-s)\beta (H_L(b) - \mu \,\mathcal{N}_L)}\biggr)
    \\ &= \nonumber
    \frac{1}{(2L+1)^2} \tr\bigl( \rho_L(\beta,\mu,b) \, H'_L(b) \bigr)
    \\ &= \label{eq:magnetization-as-sum-over-currents}
    \frac{1}{(2L+1)^2} \sum_{m=-L}^{L-1} \sum_{n=0}^{2L} n \tr\bigl( \rho_L(\beta,\mu,b) \, J_{1,L}^{(m,n)}(b) \bigr)
    .
\end{align}

\begin{proposition}
    \label{prop:m_L-equals-I_L^d}
    Let \(\zeta\supbulk\in \ell^1\) and \((H_L(b))_{L\in\N}\) be a family of Hamiltonians of the form~\eqref{eq:Hb} satisfying local indistinguishability of the Gibbs state in the bulk at \((\beta,\mu,b)\) with \(\zeta\supbulk\)-decay in the sense of Definition~\ref{def:local-indistinguishability}.
    Then
    \begin{equation*}
        \Bigl\lvert
            m_L(\beta,\mu,b)
            - \Iedged[L]{d}(\beta,\mu,b)
        \Bigr\rvert
        \leq
        \Cb \, \biggl(\frac{4\,d^2}{L} + \smashoperator[lr]{\sum_{n=d-R-D}^\infty} \zeta\supbulk(n) \biggr)
    \end{equation*}
    for all \(d > D+R\) with \(\Cb\) from \nameref{prop:Bloch-theorem}.
\end{proposition}

\begin{proof}
    We decompose the sum from~\eqref{eq:magnetization-as-sum-over-currents} into five regions (see Figure~\ref{fig:decomposition-into-five-regions})
    \begin{equation} \label{eq:decomposition-into-five-regions}
        \frac{1}{(2L+1)^2} \sum_{m=-L}^{L-1} \sum_{n=0}^{2L} n \tr\bigl( \rho_L(\beta,\mu,b) \, J_{1,L}^{(m,n)}(b) \bigr)
        =
        \Abulk + \Aleft + \Aright + \Abottom + \Atop.
    \end{equation}
    We will show that \(\Abulk\), \(\Aleft\) and \(\Aright\) are small and that \(\Atop\) and \(\Abottom\) resemble the edge current.
    Abbreviating \(j_{k,L}^{(m,n)} := \tr \bigl( \rho_L(\beta,\mu,b) \, J_{k,L}^{(m,n)}(b) \bigr)\), the individual contributions are
    \begin{align*}
        \Abulk &:=\frac{1}{(2L+1)^2} \sum_{m=-L+d}^{L-d-1} \sum_{n=d}^{2L-d} n \, j_{1,L}^{(m,n)}
        , \\
        \Aleft &:= \frac{1}{(2L+1)^2} \sum_{m=-L}^{-L+d-1} \sum_{n=0}^{2L} n \, j_{1,L}^{(m,n)}
        , \\
        \Aright &:= \frac{1}{(2L+1)^2} \sum_{m=L-d}^{L-1} \sum_{n=0}^{2L} n \, j_{1,L}^{(m,n)}
        , \\
        \Abottom &:= \frac{1}{(2L+1)^2} \sum_{m=-L+d}^{L-d-1} \sum_{n=0}^{d-1} n \, j_{1,L}^{(m,n)}
        ,\quad\text{and} \\
        \Atop &:= \frac{1}{(2L+1)^2} \sum_{m=-L+d}^{L-d-1} \sum_{n=2L-d+1}^{2L} n\, j_{1,L}^{(m,n)}
        ,
    \end{align*}
    where \(\Abulk\) is the bulk part, \(\Aleft\) and \(\Aright\) are the sum over the left and right edge regions, \(\Abottom\) and \(\Atop\) are the sum over the upper and lower edge regions.
    Note that \(d>R+D\).

    \begin{figure}
        \fcapside{
            \setlength{\ul}{.8cm}
            \setlength{\vertexdiameter}{\lw}
            \begin{tikzpicture}[myLine,x=\ul,y=\ul]
                \foreach \x in {-4,-3,...,4}{
                    \foreach \y in {-4,-3,...,4}{
                        \path[fill=gray] (\x,\y) circle (\vertexdiameter);
                    }
                }
                \foreach \x in {-4,-3,...,3}{
                    \foreach \y in {-4,-3,...,4}{
                        \draw[line cap=round, shorten=3.5\lw, color=gray] (\x,\y) ++(0.5,-0.5) -- +(0,1);
                    }
                }
                \begin{scope}
                    \draw[color=col_1, on background={fill, color=col_1!20}] 
                        (-4.25\ul+1.5\lw, -4.5\ul+1.5\lw) rectangle
                        (-2.25\ul-1.5\lw,  4.5\ul-1.5\lw)
                        node[midway, fill=col_1!20, inner sep=2\lw, minimum height=1\ul] {left};
                    \draw[color=col_1, on background={fill, color=col_1!20}] 
                        ( 1.75\ul+1.5\lw, -4.5\ul+1.5\lw) rectangle
                        ( 3.75\ul-1.5\lw,  4.5\ul-1.5\lw)
                        node[midway, fill=col_1!20, inner sep=2\lw, minimum height=1\ul] {right};
                    \draw[color=col_2, on background={fill, color=col_2!20}] 
                        (-2.25\ul+1.5\lw, -4.5\ul+1.5\lw) rectangle
                        ( 1.75\ul-1.5\lw, -2.5\ul-1.5\lw)
                        node[midway, fill=col_2!20, inner sep=2\lw] {bottom};
                    \draw[color=col_2, on background={fill, color=col_2!20}] 
                        (-2.25\ul+1.5\lw,  4.5\ul-1.5\lw) rectangle
                        ( 1.75\ul-1.5\lw,  2.5\ul+1.5\lw)
                        node[midway, fill=col_2!20, inner sep=2\lw] {top};
                    \draw[color=col_3, on background={fill, color=col_3!20}] 
                        (-2.25\ul+1.5\lw, -2.5\ul+1.5\lw) rectangle
                        ( 1.75\ul-1.5\lw,  2.5\ul-1.5\lw)
                        node[midway, fill=col_3!20, inner sep=2\lw, minimum height=1\ul] {bulk};
                \end{scope}
            \end{tikzpicture}
        }{
            \caption{
                Depicted are \(\Lambda_L\) for \(L=4\) (dots) and the corresponding dual edges \(e_{1,(m,n)}\) (lines) as defined above~\eqref{eq:current-through-dual}. 
                They were used to define the current \(J_{1,L}^{(m,n)}(b)\) in~\eqref{eq:current-through-dual}.
                Since there are less vertical dual edges than vertices in \(\Lambda_L\), the right most vertices have no corresponding dual edge.
                The coloured boxes group the dual edges into the five groups of the decomposition~\eqref{eq:decomposition-into-five-regions} for \(d=2\).
            }
            \label{fig:decomposition-into-five-regions}
        }
    \end{figure}
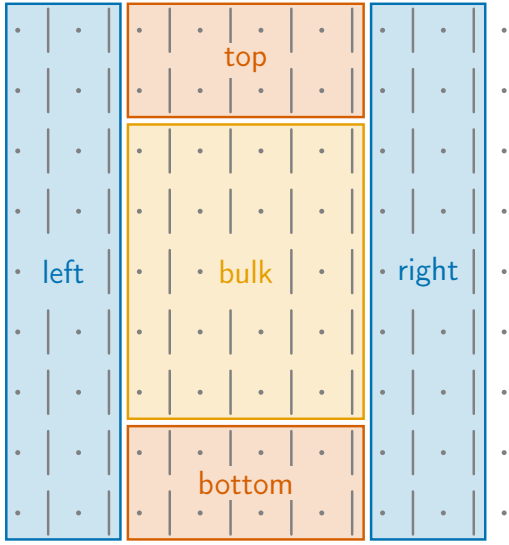

    By \nameref{prop:Bloch-theorem}, \(\lvert j_{1,L}^{z} \rvert \leq \zeta(\dist(z,\Z^2\setminus\Lambda_L))\) with \(\zeta(r) := \Cb \, \zeta\supbulk(r-D-R)\).
    Thus, we can bound the inner part \(\Abulk\) by summing over shells \(\{\, z\in\Lambda_L \,|\, \dist(z,\Z^2\setminus\Lambda_L)=r \,\}\)
    \begin{align}
        \bigl\lvert \Abulk \bigr\rvert
        & \leq
        \frac{1}{2L+1} \sum_{m=-L+d}^{L-d-1} \sum_{n=d}^{2L-d} \bigl\lvert j_{1,L}^{(m,n)} \bigr\rvert
        \leq \frac{L}{2L+1} \sum_{r=d}^{L}
        \zeta(r)
        .
        \label{eq:bound-inner-part}
    \end{align}

    Then, let us consider the right and left edge regions, \(\Aleft\) and \(\Aright\).
    We write only the case of \(\Aright\), since the other one is analogous.
    First, using a discrete version of partial integration, we find
    \begin{equation*}
        \sum_{m=L-d}^{L-1} \sum_{n=0}^{2L} n \, j_{1,L}^{(m,n)}
        = \sum_{m=L-d}^{L-1} \sum_{n=1}^{2L} \sum_{n'=n}^{2L} j_{1,L}^{(m,n')}
    \end{equation*}
    and by current conservation (see edge set~(c) in Figure~\ref{fig:current-conservation}) and \(\big\lvert j_{2,L}^{(m',n)} \big\rvert < \Cb\),
    \begin{equation*}
        \Biggl\lvert \sum_{n'=n}^{2L} j_{1,L}^{(m,n')} \Biggr\rvert
        =
        \Biggl\lvert \sum_{m'=m+1}^L j_{2,L}^{(m',n-1)} \Biggr\rvert
        \leq
        (L-m) \, \Cb.
    \end{equation*}

    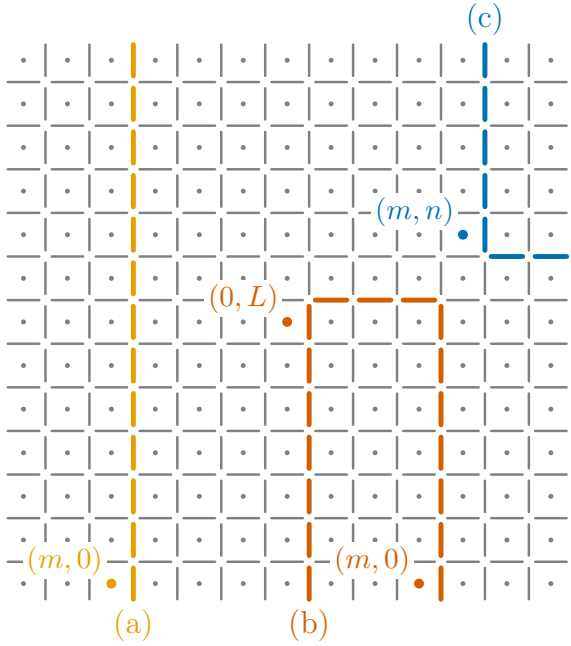
\begin{figure}
        \fcapside[.47\textwidth]{
            \newcommand\LL{6}
            \setlength{\vertexdiameter}{2\lw}
            \setlength{\ul}{(\linewidth+3.5\lw)/13}  
            \newcommand{\edgex}[2][gray]{\draw[line cap=round, color=#1] (#2) ++(0.5,-0.5\ul+2.5\lw) -- +(0,1\ul-5\lw)}  
            \newcommand{\edgey}[2][gray]{\draw[line cap=round, color=#1] (#2) ++(-0.5\ul+2.5\lw,0.5) -- +(1\ul-5\lw,0)}  
            \newcommand{\vertex}[2][gray]{\path[fill=#1] (#2) circle[radius=.5\vertexdiameter]}
            \tikzset{edgeset label/.style={inner sep=1.5\lw,outer sep=0pt,font=\rmfamily}}
            \begin{tikzpicture}[myLine,x=\ul,y=\ul]
                \foreach \x in {-\LL,...,\LL}{
                    \foreach \y in {-\LL,...,\LL}{
                        \vertex{\x,\y};
                    }
                }
                \foreach \x [parse=true] in {-\LL,...,\LL-1}{
                    \foreach \y in {-\LL,...,\LL}{
                        \edgex{\x,\y};
                    }
                }
                \foreach \x in {-\LL,...,\LL}{
                    \foreach \y [parse=true] in {-\LL,...,\LL-1}{
                        \edgey{\x,\y};
                    }
                }
                \begin{scope}[myLine=2]
                    \setlength{\vertexdiameter}{4\lw}
                    \newcommand{\rightm}{4}
                    \newcommand{\rightn}{2}
                    \foreach \y in {\rightn,...,\LL}{
                        \edgex[col_1]{\rightm,\y};
                    }
                    \foreach \x [parse=true] in {\rightm+1,...,\LL}{
                        \edgey[col_1]{\x,\rightn-1};
                    }
                    \vertex[col_1]{\rightm,\rightn} node[anchor=south east,color=col_1,fill=white,inner sep=1.5\lw, outer sep=\vertexdiameter/2] {\small\((m,n)\)};
                    \path (\rightm,\LL) ++ (0.5,0.5) node[above, color=col_1,edgeset label] {(c)};
                    \newcommand{\bottomm}{3}
                    \foreach \y in {-\LL,...,0}{
                        \edgex[col_2]{0,\y};
                        \edgex[col_2]{\bottomm,\y};
                    }
                    \foreach \x [parse=true] in {1,...,\bottomm}{
                        \edgey[col_2]{\x,0};
                    }
                    \vertex[col_2]{\bottomm,-\LL} node[anchor=south east,color=col_2,fill=white,inner sep=1.5\lw, outer sep=\vertexdiameter/2] {\small\((m,0)\)};
                    \vertex[col_2]{0,0} node[anchor=south east,color=col_2,fill=white,inner sep=1.5\lw, outer sep=\vertexdiameter/2] {\small\((0,L)\)};
                    \path (0,-\LL) ++ (0.5,-0.5) node[below, color=col_2, edgeset label] {(b)};
                    \newcommand{\linem}{-4}
                    \foreach \y in {-\LL,...,\LL}{
                        \edgex[col_3]{\linem,\y};
                    }
                    \vertex[col_3]{\linem,-\LL} node[anchor=south east,color=col_3,fill=white,inner sep=1.5\lw, outer sep=\vertexdiameter/2] {\small\((m,0)\)};
                    \path (\linem,-\LL) ++ (0.5,-0.5) node[below, color=col_3, edgeset label] {(a)};
                \end{scope}
            \end{tikzpicture}
        }{
            \captionsetup{singlelinecheck=off}
            \caption[Visualization of sets for which we use current conservation.]{
                The picture shows \(\Lambda_L\) for \(L=\LL\) and all dual edges.
                A few connected sets of dual edges are highlighted.
                Together with the boundary all these sets are rectangles, and we use current conservation through the edges in the proofs.
                More precisely,
                \begin{itemize}[align=left,leftmargin=*]
                    \color{col_3}
                    \item[(a)]\(\displaystyle
                        \sum_{n=0}^{2L} j_{1,L}^{(m,n)} = 0
                        \)\normalcolor ,
                    \color{col_2}
                    \item[(b)]\(\displaystyle
                        \sum_{n=0}^{L} \bigl( -j_{1,L}^{(0,n)} + j_{1,L}^{(m,n)} \bigr) + \sum_{m'=1}^{m} j_{2,L}^{(m',L)} = 0
                        \)\normalcolor ,
                    \color{col_1}
                    \item[(c)]\(\displaystyle
                        \sum_{n'=n}^{2L} j_{1,L}^{(m,n')} + \sum_{m'=m+1}^L j_{2,L}^{(m',n-1)} = 0
                        \)\normalcolor .
                \end{itemize}
            }
            \label{fig:current-conservation}
        }
    \end{figure}

    Hence,
    \begin{equation*}
        \bigl\lvert \Aright \bigr\rvert
        =
        \Biggl\lvert
            \frac{1}{(2L+1)^2} \sum_{m=L-d}^{L-1} \sum_{n=0}^{2L} n \, j_{1,L}^{(m,n)}
        \Biggr\rvert
        \leq
        \frac{d^2\,\Cb}{2L+1}.
    \end{equation*}

    Let us now consider the contributions coming from the lower and upper edge regions.
    The contribution \(\Abottom\) is small due to the prefactor \(n \approx 0\):
    \begin{equation*}
        \Biggl\lvert
            \frac{1}{(2L+1)^2} \sum_{m=-L+d}^{L-d-1} \sum_{n=0}^{d-1} n \, j_{1,L}^{(m,n)}
        \Biggr\rvert
        \leq \frac{1}{(2L+1)^2} \, 2(L-d) \, d \, d \, \Cb
        \leq \frac{d^2 \, \Cb}{2L+1}.
    \end{equation*}
    We have now proven that all contributions up to \(\Atop\) are small.
    In \(\Atop\) we replace \(n\) with \(2L\) making again an error of
    \begin{equation*}
        \Biggl\lvert
            \frac{1}{(2L+1)^2} \sum_{m=-L+d}^{L-d-1} \sum_{n=2L-d+1}^{2L} \bigl(n-2L\bigr) \, j_{1,L}^{(m,n)}
        \Biggr\rvert
        \leq \frac{d^2 \, \Cb}{2L+1}.
    \end{equation*}
    It now seems, that the magnetization only stems from the top part.
    That, however, is only due to our choice of the gauge.
    Indeed, the top and bottom contributions would equal in a gauge with Peierls phase \(\E^{\I\frac{x_2+y_2-2L \xmathstrut{.2}}{2}b(x_1-y_1)}\) in~\eqref{eq:definition-Tb} – which corresponds to exactly our Hamiltonian but on boxes \([-L,L]^2\) centered around the origin.

    However, also with our natural choice of gauge, we can rewrite the magnetization in terms of the currents near the bottom edge using current conservation (see edge set~(a) in Figure~\ref{fig:current-conservation}) and vanishing of the currents in the bulk:
    \begin{align*}
        &\alignindent
        \Biggl\lvert
            \frac{1}{(2L+1)^2} \sum_{m=-L+d}^{L-d-1} \smashoperator[r]{\sum_{n=2L-d+1}^{2L}} 2L \, j_{1,L}^{(m,n)}
            -
            \frac{1}{(2L+1)^2} \sum_{m=-L+d}^{L-d-1} \sum_{n=0}^{d-1} 2L \, \bigl(-j_{1,L}^{(m,n)}\bigr)
        \Biggr\rvert
        \\&\leq
        \frac{1}{2\,(2L+1)} \sum_{m=-L+d}^{L-d-1} \sum_{n=d}^{2L-d} \bigl\lvert j_{1,L}^{(m,n)} \bigr\rvert
        .
    \end{align*}
    The last step here follows exactly from the calculation~\eqref{eq:bound-inner-part} for the inner part.

    So far, we have proven that
    \begin{equation}
        \label{eq:Aux1Thm1}
        \Biggl\lvert
            m_L(\beta,\mu,b)
            -\frac{-2L}{(2L+1)^2} \sum_{m=-L+d}^{L-d-1} \sum_{n=0}^{d-1} j_{1,L}^{(m,n)}
        \Biggr\rvert
        \leq
        \frac{4\,d^2\,\Cb}{2L+1} + \sum_{n=d}^\infty\zeta(n).
    \end{equation}
    It remains to show that all contributions equal the one at \(m=0\).
    Using current conservation once more (now for edge set~(b) in Figure~\ref{fig:current-conservation}), for any \(m>0\)
    \begin{equation}
        \label{eq:Aux1Thm2}
        \Biggl\lvert
            \sum_{n=0}^{d-1} \Bigl(j_{1,L}^{(m,n)} - j_{1,L}^{(0,n)} \Bigr)
        \Biggr\rvert
        \leq
        \sum_{n=d}^{L} \Bigl( \bigl\lvert j_{1,L}^{(m,n)} \bigr\rvert + \bigl\lvert j_{1,L}^{(0,n)} \bigr\rvert \Bigr)
        + \sum_{m'=1}^{m} \bigl\lvert j_{2,L}^{(m',L)} \bigr\rvert,
    \end{equation}
    and similarly for \(m<0\).
    Thus, by using~\eqref{eq:Aux1Thm2} and the analogue estimate for \(m>0\), we can estimate the error that we make by replacing \(j_{1,L}^{(m,n)}\) with \(j_{1,L}^{(0,n)}\) in~\eqref{eq:Aux1Thm1}.
    The error has three terms coming from the estimate~\eqref{eq:Aux1Thm2} and each of them can be bounded using again the strategy from~\eqref{eq:bound-inner-part}:
    \begin{align*}
        \frac{2L}{(2L+1)^2} \sum_{m=-L+d}^{L-d-1} \sum_{n=d}^{L} \bigl\lvert j_{1,L}^{(m,n)} \bigr\rvert
        & \leq \sum_{n=d}^{\infty} \zeta(n)
        , \\
        \frac{2L}{(2L+1)^2} \sum_{m=-L+d}^{L-d-1} \sum_{n=d}^{L} \bigl\lvert j_{1,L}^{(0,n)} \bigr\rvert
        &\leq
        \sum_{n=d}^{\infty} \zeta(n)
        , \\
        \frac{2L}{(2L+1)^2} \sum_{m=-L+d}^{L-d-1} \sum_{m'=1}^{m} \bigl\lvert j_{1,L}^{(m',L)} \bigr\rvert
        &\leq
        \smashoperator[lr]{\sum_{m'=1}^{L-d-1}} \zeta(L-m')
        \leq
        \sum_{n=d}^{\infty} \zeta(n)
        .
    \end{align*}

    As a final step, we bound
    \begin{equation*}
        \Biggl\lvert
            \Bigl( \tfrac{2L \, 2(L-d)}{(2L+1)^2} - 1 \Bigr)
            \sum_{n=0}^{d-1} j_{1,L}^{(0,n)}
        \Biggr\rvert
        \leq
        \frac{4 \, d^2 \, \Cb}{2L+1}
    \end{equation*}
    Hence, altogether we have proven that
    \begin{equation*}
        \Bigl\lvert
            m_L(\beta,\mu,b)
            - \Iedged[L]{d}(\beta,\mu,b)
        \Bigr\rvert
        \leq
        \Cb \, \biggl(\frac{8\,d^2}{2L+1} + \smashoperator[lr]{\sum_{n=d-R-D}^\infty} \zeta\supbulk(n) \biggr)
        .
    \end{equation*}
\end{proof}

We conclude this section with the proof of Theorem~\ref{thm:main}\@.

\begin{proof}[Proof of Theorem~\ref{thm:main}\@]
    Combining Propositions~\ref{prop:localization-I_L} and~\ref{prop:m_L-equals-I_L^d}, gives
    \begin{equation*}
        \bigl\lvert
            m_L(\beta,\mu,b) - \Iedge[L](\beta,\mu,b)
        \bigr\rvert
        \leq
        \Cb \, \biggl(\frac{4\,d^2}{L} + 2 \smashoperator[lr]{\sum_{n=d-R-D}^\infty} \zeta\supbulk(n) \biggr)
    \end{equation*}
    for all \(d \geq R+D\).
    Taking the minimum over \(d\) gives~\eqref{eq:thm-m_L-equals-I_L} with
    \begin{equation}\label{eq:theta-def}
        \theta(L) = 2\,\Cb \min_{R+D\leq d\leq L} \biggl(\frac{2\,d^2}{L} + \smashoperator[lr]{\sum_{n=d-R-D}^\infty} \zeta\supbulk(n) \biggr),
    \end{equation}
    which tends to zero for \(L\to \infty\) (choose, e.g., \(d=L^{1/4}\)).
    The bound~\eqref{eq:thm-localization-I_L} follows directly from Proposition~\ref{prop:localization-I_L}.
\end{proof}

\subsection{Proof of Theorem~\ref{thm:limit-I}}

Theorem~\ref{thm:limit-I} is basically the thermodynamic limit version of Theorem~\ref{thm:main} plus some additional remarks.
We split the proof into three parts: in the first part we show the localization of the edge current in the thermodynamic limit, namely Proposition~\ref{prop:I_infinity-Bloch-theorem-localization-convergence}, in the second part we prove the equality with the thermodynamic limit of the magnetization, that is Proposition~\ref{prop:m-equals-I-and-convergence-m_L-equals-I}, and in the last part we show the independence of the edge interaction, see Proposition~\ref{prop:different-edges}.

Let us start with an important remark.
In section~\ref{subsec:results} we already discussed that local indistinguishability everywhere with \(\zeta\supedge\)-decay implies existence of a unique thermodynamic limit state \(\rho_\infty(\beta,\mu,b)\).
More precisely, for finite \(X\subset \Z^2_+\) and \(A\in \alg_X\) the convergence is given by
\begin{equation}
    \label{eq:convergence-tdl}
    \bigl\lvert
        \rho_\infty(\beta,\mu,b)(A)
        - \tr\bigl(\rho_L(\beta,\mu,b)\, A\bigr)
    \bigr\rvert
    \leq
    \lVert A \rVert
    \, g\bigl(\lvert X \rvert\bigr)
    \, \zeta\supedge\bigl(\dist(X, \Z^2_+ \setminus \Lambda_L)\bigr)
\end{equation}
due to local indistinguishability.

\subsubsection{Localization of the current in the thermodynamic limit}

We first note that \nameref{prop:Bloch-theorem} and the localization of the edge current carry over to the thermodynamic limit.

\begin{proposition}
    \label{prop:I_infinity-Bloch-theorem-localization-convergence}
    Let \(\zeta\supbulk\in \ell^1\), \(\zeta\supedge\) a null sequence and \((H_L(b))_{L\in\N}\) be a family of Hamiltonians of the form~\eqref{eq:Hb} satisfying local indistinguishability of the Gibbs state at \((\beta,\mu,b)\) with \(\zeta\supbulk\)-decay in the bulk and \(\zeta\supedge\)-decay everywhere, in the sense of Definition~\ref{def:local-indistinguishability}.
    Then the unique thermodynamic limit state \(\rho_\infty(\beta,\mu,b)\) on \(\alg_\mathrm{loc}\) satisfies Bloch’s theorem, that is, for all \(z\in \Z^2_+\)
    \begin{equation*}
        \bigl\lvert
            \rho_\infty(\beta,\mu,b)\bigl(J_{1}^{z}(b)\bigr)
        \bigr\rvert
        \leq
        2 \, \Cb \, \zeta\supbulk\Bigl(\bigl[\dist(z,\Z^2 \setminus \Z^2_+)-D-R\bigr]_+\Bigr),
    \end{equation*}
    with \(\Cb\) from \nameref{prop:Bloch-theorem}.
    Hence, also the infinite edge current, defined in~\eqref{eq:definition-infinite-edge-current}, is localized near the boundary in the sense that
    \begin{equation*}
        \bigl\lvert
            \Iedged{d}(\beta,\mu,b) - \Iedge(\beta,\mu,b)
        \bigr\rvert
        \leq
        2 \, \Cb \smashoperator[lr]{\sum_{n=d-D-R}^\infty} \zeta\supbulk(n)
    \end{equation*}
    for all \(d \geq D + R\).
    Moreover, for all \(L \geq R\) and \(d \leq L\)
    \begin{equation*}
        \bigl\lvert
            \Iedged[L]{d}(\beta,\mu,b) - \Iedged{d}(\beta,\mu,b)
        \bigr\rvert
        \leq
        d \, \zeta\supedge(L-R)
        .
    \end{equation*}
\end{proposition}

\begin{proof}
    Combining~\eqref{eq:convergence-tdl} with \nameref{prop:Bloch-theorem} we find for all \(z\in \Z^2_+\) and \(L\) such that \(z\in \Lambda_{L-R}\) (remember that then \(J_{1}^{z}(b)=J_{1,L}^{z}(b)\)) that
    \begin{align*}
        &\alignindent
        \bigl\lvert
            \rho_\infty(\beta,\mu,b)\bigl(J_{1}^{z}(b)\bigr)
        \bigr\rvert
        \\&\leq
        \bigl\lvert
            \rho_\infty(\beta,\mu,b)\bigl(J_{1,L}^{z}(b)\bigr) - \tr\bigl(\rho_L(\beta,\mu,b)\, J_{1,L}^{z}(b)\bigr)
        \bigr\rvert
        +
        \bigl\lvert
            \tr\bigl(\rho_L(\beta,\mu,b)\, J_{1,L}^{z}(b)\bigr)
        \bigr\rvert
        \\&\leq
        \zeta\supedge\bigl(\dist(z, \Z^2_+ \setminus \Lambda_L)-R\bigr)
        + \Cb \, \zeta\supbulk\Bigl(\bigl[\dist(z,\Z^2 \setminus \Lambda_L)-D-R\bigr]_+\Bigr)
        ,
    \end{align*}
    where the first quantity converges to zero, and \(\dist(z,\Z^2 \setminus \Lambda_L)=\dist(z,\Z^2 \setminus \Z^2_+)\) for \(L\) large enough.
    The localization of the current now follows exactly as in Proposition~\ref{prop:localization-I_L}.
    For the convergence in \(L\), we estimate
    \begin{align*}
        \bigl\lvert
            \Iedged[L]{d}(\beta,\mu,b) - \Iedged{d}(\beta,\mu,b)
        \bigr\rvert
        &\leq
        \sum_{n=0}^{d-1}
        \, \Bigl\lvert
            \tr\Bigl( \rho_L(\beta,\mu,b) \, J_{1}^{(0,n)}(b) \Bigr)
            - \rho_\infty(\beta,\mu,b) \Bigl( J_{1}^{(0,n)}(b) \Bigr)
        \Bigr\rvert
        \\&\leq
        d \, \zeta\supedge(L-R)
        .\qedhere
    \end{align*}
\end{proof}

\subsubsection{Magnetization in the thermodynamic limit}
\label{sec:proof-theorem-limit-I}

\begin{proposition}
    \label{prop:m-equals-I-and-convergence-m_L-equals-I}
    Let \(\zeta\supbulk \in \ell^1\) and \(\zeta\supedge\) tend to zero.
    If \((H_L(b))_{L\in\N}\) is a family of Hamiltonians of the form~\eqref{eq:Hb} satisfying local indistinguishability of the Gibbs state at \((\beta,\mu,b)\) with \(\zeta\supbulk\)-decay in the bulk and \(\zeta\supedge\)-decay everywhere, in the sense of Definition~\ref{def:local-indistinguishability}, then
    \begin{equation*}
        \lim_{L\to\infty} m_L(\beta,\mu,b)
        =
        \Iedge(\beta,\mu,b)
        .
    \end{equation*}
\end{proposition}

\begin{proof}
    We combine Propositions~\ref{prop:m_L-equals-I_L^d} and~\ref{prop:I_infinity-Bloch-theorem-localization-convergence} to obtain
    \begin{equation}
        \label{eq:proof-explicit-bound-m_L-equals-I}
        \bigl\lvert
            m_L(\beta,\mu,b) - \Iedge(\beta,\mu,b)
        \bigr\rvert
        \leq
        \Cb \, \biggl(\frac{4\,d^2}{L} + \smashoperator[lr]{\sum_{n=d-R-D}^\infty} \zeta\supbulk(n) \biggr)
        + \, d \, \zeta\supedge(L-R),
    \end{equation}
    which will tend to zero as \(L\to \infty\) for an appropriate choice of \(d\), depending on \(L\).
\end{proof}

\subsubsection{Independence of the specific edge Hamiltonian\texorpdfstring{~\(T_b\supedge\) and~\(\Phi\supedge\)}{}}

In order to complete the proof of Theorem~\ref{thm:limit-I} it remains to show that the thermodynamic limit of the magnetization and the current does not depend on the specific edge contributions~\(T_b\supedge\) and~\(\Phi\supedge\).
We prove this by showing that the edge currents of two finite systems with different edge Hamiltonians are asymptotically equivalent and then taking the thermodynamic limit.

Taking into account the differentiability result with respect to \(\mu\) of Theorem~\ref{thm:muderivative}, this also implies the insensitivity to boundary perturbations of the \(\mu\)-derivatives in the thermodynamic limit.

\begin{proposition}
    \label{prop:different-edges}
    Let \(\zeta\supbulk\in \ell^1\), \(\zeta\supedge\) tend to zero, and let \((H_L(b))_{L\in\N}\) and \((\tilde H_L(b))_{L\in\N}\) be two families of Hamiltonians of the form~\eqref{eq:Hb} that only differ in the definition of the edge contributions \(T_b\supedge\), \(\Phi\supedge\) and \(\tilde T_b\supedge\), \(\tilde \Phi\supedge\).
    Assume that both satisfy local indistinguishability of the Gibbs state at \((\beta,\mu,b)\) with \(\zeta\supbulk\)-decay in the bulk and \(\zeta\supedge\)-decay everywhere in the sense of Definition~\ref{def:local-indistinguishability}.
    Then, denoting the quantities of \(\tilde H_L(b)\) with an additional tilde,
    \begin{equation*}
        \Iedge(\beta,\mu,b)
        =
        \tIedge(\beta,\mu,b)
        .
    \end{equation*}
\end{proposition}

This result shows, that a local perturbation near the edge might change, where exactly the edge current flows, but it does not change the total net current near the boundary.
In this sense, the edge current \(\Iedge(\beta,\mu,b)\) is dictated by the bulk.

\begin{proof}
    The idea of the proof is to use current conservation and \nameref{prop:Bloch-theorem} to show that the current along the lower and upper edge are the same up to a sign (for both interactions).
    Using local indistinguishability we can then prove that the currents along the upper edge are almost the same for both edge interactions.
    Thus, also the currents along the lower edge almost agree.

    As in section~\ref{sec:proof-theorem-m_L-equals-I_L}, we abbreviate \(j_{k,L}^{(m,n)} := \tr \bigl( \rho_L(\beta,\mu,b) \, J_{k,L}^{(m,n)}(b) \bigr)\).
    And for better readability, we drop the arguments \((\beta,\mu,b)\) in the following.
    The current flowing in the upper edge of the box can be written as
    \begin{equation*}
        \Iedged[L]{{d\,{\mathrm{up.}}}}:= \sum_{n=0}^{d-1} j_{1,L}^{(0,2L-n)}.
    \end{equation*}
    By current conservation (see edge set~(a) in Figure~\ref{fig:current-conservation}) and \nameref{prop:Bloch-theorem} we find
    \begin{equation*}
        \Bigl\lvert
            \Iedged[L]{d}
            + \Iedged[L]{{d\,{\mathrm{up.}}}}
        \Bigr\rvert
        \leq
        \sum_{n=d}^{2L-d} \bigl\lvert j_{1,L}^{(0,n)} \bigr\rvert
        \leq
        2 \, \Cb \smashoperator[lr]{\sum_{n=d-R-D}^\infty} \zeta\supbulk(n)
    \end{equation*}
    for all \(L \geq d \geq D+R\).
    Then, using local indistinguishability for \(\Lambda'=B_L^{(0,2L-n)}(2L-n-D)\), with \(0<n<d-1\), which is chosen such that \(T_b\supedge\) and \(\Phi\supedge\) vanish on \(\Lambda'\), we find
    \begin{equation*}
        \Bigl\lvert
            j_{1,L}^{(0,2L-n)}
            - \tr \Bigl(\rho\bigl[H_L\big|_{\Lambda'}\bigr] \, J_{1,L}^{(0,2L-n)} \Bigr)
        \Bigr\rvert
        \leq
        \zeta\supedge( 2L-n-R-D ).
    \end{equation*}
    And the same also holds for the corresponding properties of \(\tilde H_L\) denoted by an additional tilde.
    Due to the choice of \(\Lambda'\), \(H_L\big|_{\Lambda'} = \tilde H_L\big|_{\Lambda'}\), and we also have
    \(J_{1,L}^{(0,2L-n)} = \tilde J_{1,L}^{(0,2L-n)}\) for \(2L-n > D + R\).
    Hence,
    \begin{align}
        \bigl\lvert
            \Iedged[L]{d} - \tIedged[L]{d}
        \bigr\rvert
        &\leq
        \notag
        \bigl\lvert
            \Iedged[L]{{d\,{\mathrm{up.}}}} - \tIedged[L]{{d\,{\mathrm{up.}}}}
        \bigr\rvert
        + \bigl\lvert \Iedged[L]{d} + \Iedged[L]{{d\,{\mathrm{up.}}}} \bigr\rvert
        + \bigl\lvert \tIedged[L]{{d\,{\mathrm{up.}}}} + \tIedged[L]{d} \bigr\rvert
        \\&\leq
        \notag
        \sum_{n=0}^{d-1} \Bigl\lvert
            j_{1,L}^{(0,2L-n)} - {\tilde \jmath}_{1,L}^{(0,2L-n)}
        \Bigr\rvert
        + 4 \, \Cb \smashoperator[lr]{\sum_{n=d-R-D}^\infty} \zeta\supbulk(n)
        \\&\leq
        \label{eq:boundary_independence}
        2 \, d \, \zeta\supedge(2L+1-d-D-R)
        + 4 \, \Cb \smashoperator[lr]{\sum_{n=d-R-D}^\infty} \zeta\supbulk(n)
        .
    \end{align}
    The statement now follows from Proposition~\ref{prop:I_infinity-Bloch-theorem-localization-convergence}.
\end{proof}

\subsection{Proof of Theorem~\ref{thm:muderivative}}
\label{sec:proof-theorem-muderivative}

In this section, we discuss the \(\mu\)-derivative of the edge current.
Defining
\begin{equation}
    \label{eq:definition-mu-derivative-state}
    \mathcal{F}_L(\beta,\mu,b)
    := \dmu \, \rho_L(\beta,\mu,b)
    =
    \beta \, \bigl(\mathcal{N}_L- \langle \mathcal{N}_L\rangle_{\rho_L(\beta,\mu,b)} \bigr) \, \rho_L(\beta,\mu,b)
    ,
\end{equation}
we find
\begin{equation}
    \label{eq:mu-derivative-local-current}
    \dmu \tr \bigl( \rho_L(\beta,\mu,b) \, J_{1,L}^{z}(b)\bigr)
    =
    \tr \bigl( \mathcal{F}_L(\beta,\mu,b) \, J_{1,L}^{z}(b)\bigr)
\end{equation}
because \(J_{1,L}^{z}(b)\) does not depend on \(\mu\).

To prove Theorem~\ref{thm:muderivative} we use the same strategy of the proofs of Theorems~\ref{thm:main} and~\ref{thm:limit-I}, whose main ingredient is \nameref{prop:Bloch-theorem}.
Therefore, we start by proving a similar statement to \nameref{prop:Bloch-theorem} for the “state” \(\mathcal{F}_L(\beta,\mu,b)\).
A key point in the proof of \nameref{prop:Bloch-theorem} is the use of local indistinguishability for the Gibbs state.
However, since it is not clear whether \(\mathcal{F}_L(\beta,\mu,b)\) satisfies local indistinguishability, we need to adapt the strategy a bit.

Let us start by showing that \(\mathcal{F}_L(\beta,\mu,b)\) has a thermodynamic limit.
As a starting point, note that local indistinguishability everywhere implies decay of correlations, a property needed in the following.

\begin{lemma}\label{lem:decay-of-correlations}
    Let \((H_L(b))_{L\in\N}\) be a family of Hamiltonians of the form~\eqref{eq:Hb} satisfying local indistinguishability of the Gibbs state everywhere at \((\beta,\mu,b)\) with \(\zeta\supedge\)-decay in the sense of Definition~\ref{def:local-indistinguishability}.

    Then the Gibbs state satisfies decay of correlations, that is for all \(X,Y\subset \Lambda_L\), with \(\dist(X,Y) > R\), \(A\in \alg_X\) and \(B\in \alg_Y\) the covariance
    \begin{gather*}
        \operatorname{Cov}_{\rho_L(\beta,\mu,b)}(A,B)
        :=
        \tr \bigl( \rho_L(\beta,\mu,b) \, A \, B \bigr)
        - \tr \bigl( \rho_L(\beta,\mu,b) \, A \bigr)
        \tr \bigl( \rho_L(\beta,\mu,b) \, B \bigr)
    \shortintertext{is bounded by}
        \bigl\lvert \operatorname{Cov}_{\rho_L(\beta,\mu,b)}(A,B) \bigr\rvert
        \leq
        3
        \, g\bigl( \lvert X \rvert + \lvert Y \rvert \bigr)
        \, \lVert A \rVert
        \, \lVert B \rVert
        \, \zeta\supedge\left(\left\lfloor\frac{\dist(X,Y)-R-1}{2}\right\rfloor\right).
    \end{gather*}
\end{lemma}

\begin{proof}
    Let \(\ell = \bigl\lfloor\bigl(\dist(X,Y)-R-1\bigr)/2\bigr\rfloor\) and \(\Lambda' = X_\ell \cup Y_\ell\), with the \(\ell\)-neighbourhoods \(X_\ell := \bigl\{\, z\in \Lambda_L \,\big|\, \dist(z,X) \leq \ell \,\bigr\}\) and \(Y_\ell\).
    For better readability we again drop the arguments \((\beta,\mu,b)\).
    Then, for all \(Q\in \alg_{X\cup Y}\)
    \begin{equation*}
        \Bigl\lvert
            \tr \bigl( \rho_L \, Q \bigr)
            - \tr \bigl( \rho\bigl[H_L|_{\Lambda'}\bigr] \, Q \bigr)
        \Bigr\rvert
        \leq
        g\bigl( \lvert X \rvert + \lvert Y \rvert \bigr)
        \, \lVert Q \rVert
        \, \zeta\supedge(\ell).
    \end{equation*}
    Moreover, \(\rho\bigl[H_L|_{\Lambda'}\bigr] = \rho\bigl[H_L|_{X_\ell}\bigr] \otimes \rho\bigl[H_L|_{Y_\ell}\bigr]\), since \(\dist(X_\ell,Y_\ell) > R\), and thus \(\tr \bigl( \rho\bigl[H_L|_{\Lambda'}\bigr] \, A \, B \bigr) = \tr \bigl( \rho\bigl[H_L|_{\Lambda'}\bigr] \, A \bigr) \, \tr \bigl( \rho\bigl[H_L|_{\Lambda'}\bigr] \, B \bigr)\).
    The statement then follows from the triangle inequality.
\end{proof}

This allows us to prove convergence in the thermodynamic limit of the expectation value in the state \(\mathcal{F}_L(\beta,\mu,b)\) of observables that may be supported near the boundary.
We denote by
\begin{equation*}
    B_+^x(\ell) := \bigl\{\, y\in \Z^2_+ \,\vert\, \dist(x,y) \leq \ell \,\bigr\}
\end{equation*}
the ball in \(\Z^2_+\), similar to how we denoted with \(B_L^x(\ell)\) and \(B^x(\ell)\) the balls in \(\Lambda_L\) and \(\Z^2\), respectively.

\begin{lemma} \label{lem:convergence-F-L}
    Let \(n \mapsto n^2 \, \zeta\supedge(n)\in \ell^1\) and \((H_L(b))_{L\in\N}\) be a family of Hamiltonians of the form~\eqref{eq:Hb}.
    There exists a non-increasing \(\xi\in \ell^1\), explicitly given in~\eqref{eq:def-zeta-dmu}, such that the following holds:
    If \((H_L(b))_{L\in\N}\) satisfies local indistinguishability of the Gibbs state everywhere at \((\beta,\mu,b)\) with \(\zeta\supedge\)-decay in the sense of Definition~\ref{def:local-indistinguishability}, then, for all \(x\in \Z^2_+\), \(A\in \alg_{B^x_+(R)}\) and \(L'<L\) such that \(B^x_+(R)\subset \Lambda_{L'}\), it holds that
    \begin{equation} \label{eq:lem-convergence-F-L}
        \Bigl\lvert
            \tr \bigl( \mathcal{F}_L(\beta,\mu,b) \, A \bigr)
            - \tr \bigl( \mathcal{F}_{L'}(\beta,\mu,b) \, A \bigr)
        \Bigr\rvert
        \leq
        \beta \, \lVert A \rVert \, \xi\bigl(\dist(x,\Lambda_L\setminus\Lambda_{L'})-R\bigr)
        .
    \end{equation}
\end{lemma}

\begin{proof}
    For the proof let \(X=B^x_+(R)\), \(\ell=\lfloor\dist(X,\Lambda_L\setminus\Lambda_{L'})/3\rfloor\) and \(Z=B^x_+(R+2\ell)\subset \Lambda_{L'}\).
    The main idea of the proof is to write the number operators as sums of single-site operators.
    Then we can use local indistinguishability (for the single-site operators supported in \(Z\)) and decay of correlations (for the single-site operators supported outside \(Z\)) of the Gibbs state to conclude the result.
    For better readability we drop the arguments \((\beta,\mu,b)\).
    Then
    \begin{align}
        &\alignindent \nonumber
        \Bigl(
            \tr \bigl( \mathcal{F}_L \, A \bigr)
            - \tr \bigl( \mathcal{F}_{L'} \, A \bigr)
        \Bigr) / \beta
        \\ & \label{eq:proof-convergence-F-L-many-sums}
        \begin{aligned}
            ={}&
            \sum_{z\in Z}
            \tr \bigl( \rho_L \, A \, \mathcal{N}_z \bigr)
            - \tr \bigl( \rho_{L'} \, A \, \mathcal{N}_z \bigr)
            \\&
            - \sum_{z\in Z}
            \tr \bigl( \rho_L \, A \bigr)
            \tr \bigl( \rho_L \, \mathcal{N}_z \bigr)
            - \tr \bigl( \rho_L \, A \bigr)
            \tr \bigl( \rho_{L'} \, \mathcal{N}_z \bigr)
            \\&
            + \sum_{z\in Z}
            \tr \bigl( \rho_{L'} \, A \bigr)
            \tr \bigl( \rho_{L'} \, \mathcal{N}_z \bigr)
            - \tr \bigl( \rho_L \, A \bigr)
            \tr \bigl( \rho_{L'} \, \mathcal{N}_z \bigr)
            \\&
            + \sum_{z\in \Lambda_L \setminus Z}
            \tr \bigl( \rho_L \, A \, \mathcal{N}_z \bigr)
            - \tr \bigl( \rho_L \, A \bigr)
            \tr \bigl( \rho_L \, \mathcal{N}_z \bigr)
            \\&
            - \sum_{z\in \Lambda_{L'} \setminus Z}
            \tr \bigl( \rho_{L'} \, A \, \mathcal{N}_z \bigr)
            - \tr \bigl( \rho_{L'} \, A \bigr)
            \tr \bigl( \rho_{L'} \, \mathcal{N}_z \bigr)
            .
        \end{aligned}
    \end{align}

    Since \(A\) has bounded support, the first three sums can now be bounded using local indistinguishability of the Gibbs state because \(\rho_{L'}=\rho\bigl[H_L\rvert_{\Lambda_{L'}}\bigr]\), while the last two can be bounded using decay of correlations in the form of Lemma~\ref{lem:decay-of-correlations}.
    Using \(\lVert \mathcal{N}_z \rVert = 1\) and \(g(\lvert X \rvert + 1)\leq 1\), we find
    \begin{align*}
        &\alignindent
        \Bigl\lvert
            \tr \bigl( \mathcal{F}_L \, A \bigr)
            - \tr \bigl( \mathcal{F}_{L'} \, A \bigr)
        \Bigr\rvert
        \Big/ \bigl( \beta \, \lVert A \rVert \bigr)
        \\&\leq
        \begin{aligned}[t]
            & 3 \, \smashoperator[lr]{\sum_{z\in Z}}
            \, \zeta\supedge\bigl( \dist(X\cup\{z\},\Lambda_L\setminus\Lambda_{L'}) \bigr)
            \\&
            + \smashoperator[lr]{\sum_{z\in \Lambda_L \setminus Z}}
            \, 3 \, \zeta\supedge\left(\left\lfloor\frac{\dist(X,z)-R-1}{2}\right\rfloor\right)
            \\&
            + \smashoperator[lr]{\sum_{z\in \Lambda_{L'} \setminus Z}}
            \, 3 \, \zeta\supedge\left(\left\lfloor\frac{\dist(X,z)-R-1}{2}\right\rfloor\right)
            .
        \end{aligned}
    \end{align*}
    Summing over shells as in~\eqref{eq:bound-inner-part}, the sum over \(z\in Z\) is bounded by
    \begin{equation*}
        \sum_{z\in Z}
        \zeta\supedge\bigl( \dist(X\cup\{z\},\Lambda_L\setminus\Lambda_{L'}) \bigr)
        \leq
        4 \, (R+2\ell) \sum_{n=\ell}^{2\ell} \zeta\supedge(n)
        + (2R+1)^2 \, \zeta\supedge(2\ell),
    \end{equation*}
    and each of the other sums can be bounded by
    \begin{equation*}
        12 \, \smashoperator[lr]{\sum_{n=2\ell+1}^\infty} (n+R) \, \zeta\supedge\bigl(\bigl\lfloor(n-R-1)/2\bigr\rfloor\bigr)
        \leq
        12 \, \smashoperator[lr]{\sum_{m=\ell}^\infty} (2m+2+R) \, \zeta\supedge\bigl(m-\bigl\lfloor(R-1)/2\bigr\rfloor\bigr)
        .
    \end{equation*}
    Defining
    \begin{equation}
        \label{eq:def-zeta-dmu}
        \xi(r)
        =
        3 \, \Bigl(
            12 \, \smashoperator[lr]{\sum_{m=\ell}^\infty} (2m+2+R) \, \zeta\supedge\bigl(m-\bigl\lfloor(R-1)/2\bigr\rfloor\bigr)
            + (2R+1)^2 \, \zeta\supedge(2\ell)
        \Bigr)
        \Big\rvert_{\ell=\lfloor r/3 \rfloor},
    \end{equation}
    which is in \(\ell^1\) since \(n \mapsto n^2 \, \zeta\supedge(n) \in \ell^1\), yields the claim.
\end{proof}

We will use this Lemma later to prove differentiability of \(\Iedge\).
But first, we adjust the proof to show that also the \(\mu\)-derivative of the local bond currents decay in the bulk, i.e.\ a \nameref{prop:Bloch-theorem} for \(\mathcal{F}_L(\beta,\mu,b)\).
However, the two simplest approaches to adjust the proof of \nameref*{prop:Bloch-theorem} do not work.
At first, one might want to use local indistinguishability of the Gibbs state and view \(\mathcal{N}_L - \langle \mathcal{N}_L\rangle\) as part of the operator in~\eqref{eq:def-local-indistinguishability-everywhere}.
This fails, since \(\mathcal{N}_L - \langle \mathcal{N}_L\rangle\) is supported over all \(\Lambda_L\) and not bounded uniformly in \(L\).
Alternatively, one could try to use local indistinguishability of the “state” \(\mathcal{F}_L(\beta,\mu,b)\) and follow the proof of the \nameref{prop:Bloch-theorem} afterwards.
And while Lemma~\ref{lem:convergence-F-L} already looks similar to local indistinguishability, it can only be used to compare \(\dmu\, \rho_L(\beta,\mu,b)\) with \(\mathcal{F}\bigl[H_L(b)\big|_{\Lambda'}\bigr](\beta,\mu) := \dmu\, \rho\bigl[H_L(b)\big|_{\Lambda'}\bigr](\beta,\mu)\) for \(\Lambda'=\Lambda_{L'}\).
But within the proof of \nameref{prop:Bloch-theorem} we also need to compare to more general sets \(\Lambda' \subset \Lambda_L\), in particular to sets that include an edge of \(\Lambda_L\) (not only the lower one).

One might hope to adapt the proof of Lemma~\ref{lem:convergence-F-L} to prove local indistinguishability.
However, that needs decay of correlations in \(\mathcal{F}\bigl[H_L(b)\big|_{\Lambda'}\bigr](\beta,\mu)\) which we could not prove.
It would follow from local indistinguishability everywhere of \(\rho\bigl[H_L(b)\big|_{\Lambda'}\bigr](\beta,\mu)\), which might be a viable assumption since the Hamiltonian is translation invariant.
To avoid these more general assumptions, we take a different approach for which we introduce
\begin{equation} \label{eq:definition-F-Z-L}
    \mathcal{F}^Z\bigl[H_L(b)\big|_{\Lambda'}\bigr](\beta,\mu)
    :=
    \beta \, \bigl(\mathcal{N}_Z- \langle \mathcal{N}_Z\rangle_{\rho[H_L(b)|_{\Lambda'}](\beta,\mu)} \bigr) \, \rho\bigl[H_L(b)\big|_{\Lambda'}\bigr](\beta,\mu)
\end{equation}
for \(Z\subset \Lambda'\subset \Lambda_L\).
These “states” can be handled easily since the problematic sum is not present.
We prove the following statement which is similar to local indistinguishability in the state \(\mathcal{F}_L(\beta,\mu,b)\).

\begin{lemma}\label{lem:local-indistinguishability-F}
    Let \(n \mapsto n^2 \, \zeta\supedge(n)\in \ell^1\), \((H_L(b))_{L\in\N}\) be a family of Hamiltonians of the form~\eqref{eq:Hb} and \(\xi\in \ell^1\) as in Lemma~\ref{lem:convergence-F-L}.
    If \((H_L(b))_{L\in\N}\) satisfies local indistinguishability of the Gibbs state everywhere at \((\beta,\mu,b)\) with \(\zeta\supedge\)-decay in the sense of Definition~\ref{def:local-indistinguishability}, then, for all \(x\in \Lambda_L\), \(\ell \in \N_0\), \(Z=B^x_L(R+2\ell)\), \(\Lambda'=B^x_L(R+3\ell)\) and \(A\in \alg_{B^x_L(R)}\)
    \begin{equation*}
        \Bigl\lvert
            \tr \bigl( \mathcal{F}_L(\beta,\mu,b) \, A \bigr)
            - \tr \bigl( \mathcal{F}^Z\bigl[H_L(b)\big|_{\Lambda'}\bigr](\beta,\mu) \, A \bigr)
        \Bigr\rvert
        \leq
        \beta \, \lVert A \rVert \, \xi(3\ell)
        .
    \end{equation*}
\end{lemma}

\begin{proof}
    The proof is exactly the same as the proof of Lemma~\ref{lem:convergence-F-L} but with \(\rho_{L'}\) replaced by \(\rho\bigl[H_L\big|_{\Lambda'}\bigr]\) and without the last sum in~\eqref{eq:proof-convergence-F-L-many-sums}.
    Hence, we only need decay of correlations in \(\rho_L\), which is provided by Lemma~\ref{lem:decay-of-correlations}.
\end{proof}

Lemma~\ref{lem:local-indistinguishability-F} allows us to prove an analogous statement to \nameref{prop:Bloch-theorem} for the “state” \(\mathcal{F}_L(\beta,\mu,b)\), namely that \(\bigl\lvert \tr \bigl( \mathcal{F}_L(\beta,\mu,b) \, J_{k,L}^{(m,n)}(b) \bigr) \bigr\rvert\) decays rapidly inside the bulk.
In particular, we also prove boundedness of the \(\mu\)-derivative of local bond currents, which is not clear a~priori because \(\mathcal{N}_L - \langle \mathcal{N}_L\rangle\) is not uniformly bounded in~\(L\).

\begin{proposition} \label{prop:vanishing-bulk-current-derivative}
    Let \(n \mapsto n^2 \, \zeta\supedge(n)\in \ell^1\), \((H_L(b))_{L\in\N}\) be a family of Hamiltonians of the form~\eqref{eq:Hb} and \(\xi\in \ell^1\) as in Lemma~\ref{lem:convergence-F-L}.
    There exists \(\Cm>0\) such that the following holds:
    If \((H_L(b))_{L\in\N}\) satisfies local indistinguishability of the Gibbs state everywhere at \((\beta,\mu,b)\) with \(\zeta\supedge\)-decay in the sense of Definition~\ref{def:local-indistinguishability}, then
    \begin{equation*}
        \bigl\lvert \tr \bigl( \mathcal{F}_L(\beta,\mu,b) \, J_{k,L}^{z}(b) \bigr) \bigr\rvert
        \leq
        \Cm \, \beta \, \xi \Bigl(\bigl[\dist(z,\Z^2\setminus\Lambda_L)-D-R\bigr]_+\Bigr)
        .
    \end{equation*}
\end{proposition}

\begin{proof}
    For better readability, we drop the arguments \((\beta,\mu,b)\) in the proof.
    And as in the proof of \nameref*{prop:Bloch-theorem} in section~\ref{sec:proof-Bloch-theorem} we only do the proof for \(k=1\), i.e.\ currents in \(x_1\)-direction.

    We first prove uniform boundedness of the left-hand side.
    Recall that \(J_{k,L}^{z}\in \alg_{B^z_L(R)}\) and \(\lVert J_{k,L}^{z} \rVert \leq C_J\).
    For \(\ell=0\), Lemma~\ref{lem:local-indistinguishability-F} with \(B^z_L(R)=Z=\Lambda'\) yields
    \begin{align*}
        \bigl\lvert
            \tr \bigl( \mathcal{F}_L \, J_{k,L}^{z} \bigr)
        \bigr\rvert
        &\leq
        \begin{aligned}[t]
            &\Bigl\lvert
                \tr \bigl( \mathcal{F}_L \, J_{k,L}^{z} \bigr)
                - \tr \bigl( \mathcal{F}^Z\bigl[H_L\big|_{Z}\bigr] \, J_{k,L}^{z} \bigr)
            \Bigr\rvert
            \\&+ \bigl\lvert
                \tr \bigl( \mathcal{F}^Z\bigl[H_L\big|_{Z}\bigr] \, J_{k,L}^{z} \bigr)
            \bigr\rvert
        \end{aligned}
        \\&\leq
        \beta \, C_J \, \xi(0)
        + \beta \, (2R+1)^2 \, C_J
        \\&=:
        \beta \, C_0
        .
    \end{align*}

    To prove decay in the bulk, note that \(\mathcal{F}_L\) is stationary due to stationarity of the Gibbs state and since \(H_L\) commutes with the full number operator \(\mathcal{N}_L\) appearing in the definition of \(\mathcal{F}_L\).
    Hence, we have current conservation in the “state” \(\mathcal{F}_L\), see Corollary~\ref{cor:current-conservation}, which we apply for the rectangle \(\Lambda_L\cap \{\,x_1\geq 0\,\}\).
    Then we choose \(d>D+R\) and \(0<\ell<(d-D-R)/3\) to obtain
    \begin{equation*}
        \begin{aligned}
            0 ={}&\LTadd{=} \sum_{n=0}^{2L} \tr \bigl( \mathcal{F}_L \, J_{1,L}^{(0,n)} \bigr)
            \\={}& \sum_{n=0}^{d-1}
            \tr \bigl( \mathcal{F}_L \, J_{1,L}^{(0,n)} \bigr)
            + \tr \bigl( \mathcal{F}_L \, J_{1,L}^{(0,2L-n)} \bigr)
            \\&+ \sum_{n=d}^{2L-d}
            \tr \bigl( \mathcal{F}_L \, J_{1}^{(0,n)} \bigr)
            - \tr \bigl( \mathcal{F}^{B^{(0,n)}(R+2\ell)}\bigl[H_L\big|_{B^{(0,n)}(R+3\ell)}\bigr] \, J_{1}^{(0,n)} \bigr)
            \\&+ \sum_{n=d}^{2L-d}
            \tr \bigl( \mathcal{F}^{B^{(0,n)}(R+2\ell)}\bigl[H_L\big|_{B^{(0,n)}(R+3\ell)}\bigr] \, J_{1}^{(0,n)} \bigr),
        \end{aligned}
    \end{equation*}
    which is the same decomposition as in the proof of \nameref{prop:Bloch-theorem}.
    The first sum is bounded by \(2 \, d \, \beta \, C_0\).
    With Lemma~\ref{lem:local-indistinguishability-F}, the second sum is bounded by
    \begin{equation*}
        \bigl(2(L-d)+1\bigr) \, \beta \, C_J \, \xi(3\ell)
        .
    \end{equation*}
    All terms in the third sum equal the one at \((0,L)\) due to translation invariance of the Hamiltonian.
    Thus,
    \begin{equation*}
        \Bigl\lvert
            \tr \bigl(\mathcal{F}^{B^{(0,L)}(R+2\ell)}\bigl[H_L\big|_{B^{(0,L)}(R+3\ell)}\bigr] \, J_{1}^{(0,L)}\bigr)
        \Bigr\rvert
        \leq
        \frac{C_0 \, \beta \, d}{L-d} + \beta \, C_J \, \xi(3\ell)
        \to \beta \, C_J \, \xi(3\ell)
    \end{equation*}
    for \(L\to\infty\).

    We can now relate this result back to any point \(x\in \Lambda_L\).
    Choosing \(d=\dist(x,\Z^2\setminus\Lambda_L)\), \(\ell=\bigl\lfloor(d-D-R)/3\bigr\rfloor\) we find
    \begin{align*}
        \Bigl\lvert
            \tr \bigl(\mathcal{F}_L \, J_{1}^{x}\bigr)
        \Bigr\rvert
        &\leq
        \begin{aligned}[t]
            &
            \Bigl\lvert
                \tr \bigl(\mathcal{F}_L \, J_{1}^{x}\bigr)
                - \tr \bigl(\mathcal{F}^{B^{x}(R+2\ell)}\bigl[H_L\big|_{B^{x}(R+3\ell)}\bigr] \, J_{1}^{x}\bigr)
            \Bigr\rvert
            \\&+
            \Bigl\lvert
                \tr \bigl(\mathcal{F}^{B^{(0,L)}(R+2\ell)}\bigl[H_{L}\big|_{B^{(0,L)}(R+3\ell)}\bigr] \, J_{1}^{(0,L)}\bigr)
            \Bigr\rvert
        \end{aligned}
        \\&\leq
        2 \, \beta \, C_J \, \xi(3\ell),
    \end{align*}
    which proves the claim with \(\Cm = \max \bigl\{ C_0 \, \xi(0)^{-1}, 2 \, C_J\bigr\}\), using \(\xi(k) = \xi \bigl( 3 \, \lfloor k/3 \rfloor \bigr)\) due to the explicit form~\eqref{eq:def-zeta-dmu}.
\end{proof}

We can now prove differentiability of \(\Iedge\).

\begin{proposition}\label{prop:differentiability_I}
    Let \(n \mapsto n^2 \, \zeta\supedge(n)\in \ell^1\), \((H_L(b))_{L\in\N}\) be a family of Hamiltonians of the form~\eqref{eq:Hb}.
    If \((H_L(b))_{L\in\N}\) satisfies local indistinguishability of the Gibbs state everywhere with \(\zeta\supedge\)-decay at \((\beta,\mu,b)\) for all \(\mu\) in an open interval \(M\), then \(\Iedge(\beta,\mu,b)\), \(\Iedged{d}(\beta,\mu,b)\) are differentiable functions of \(\mu\in M\).
    Moreover, the derivative of \(\Iedge(\beta,\mu,b)\) is localized near the boundary, with the decay estimate
    \begin{equation}
        \bigl\lvert
            \dmu\, \Iedge(\beta,\mu,b) - \dmu\, \Iedged{d}(\beta,\mu,b)
        \bigr\rvert
        \leq \label{eq:prop-localization-dmu-I}
        \beta \, \Cm \, \smashoperator[lr]{\sum_{n=d-R-D}^\infty} \xi(n)
        ,
    \end{equation}
    where \(\Cm\), \(\xi\) are as in Proposition~\ref{prop:vanishing-bulk-current-derivative}.
\end{proposition}
\begin{proof}
    In finite volume, it is clear that \(\Iedged[L]{d}\) is differentiable in \(\mu\).
    We will thus use local indistinguishability, respectively Proposition~\ref{prop:vanishing-bulk-current-derivative}, to take the limit \(L\to \infty\) and then \(d\to \infty\) uniformly in \(\mu\).

    We abbreviate \(j_{1,L}^{(0,n)}(\mu) = \tr \bigl( \rho_L(\beta,\mu,b) \, J_{1,L}^{(0,n)}(b) \bigr)\) and \(j_{1}^{(0,n)}(\mu) = \rho_\infty(\beta,\mu,b) \bigl( J_{1}^{(0,n)}(b) \bigr)\) as before (the first will only be used for \(n \leq L\) so that \(J_{1,L}^{(0,n)}(b) = J_{1}^{(0,n)}(b)\)).
    First, by Lemma~\ref{lem:convergence-F-L},
    \begin{equation*}
        \dmu\, j_{1,L}^{(0,n)}(\mu)
        =
        \tr\Bigl( \mathcal{F}_L(\beta, \mu,b) \, J_{1}^{(0,n)}(b) \Bigr)
    \end{equation*}
    is a Cauchy sequence in \(L\).
    Denoting its limit by \(c_1^{(0,n)}(\mu)\), we have
    \begin{equation*}
        \bigl\lvert
            c_1^{(0,n)}(\mu) - \dmu\, j^{(0,n)}_{1,L}(\mu)
        \bigr\rvert
        \leq
        \beta \, C_J \, \xi(L-R)
        .
    \end{equation*}
    By local indistinguishability we know that \(j_{1,L}^{(0,n)}(\mu)\) converges to its limit \(j_{1}^{(0,n)}(\mu)\), and this convergence is uniform since \(\zeta\supedge\) is independent of \(\mu\).
    Hence, by completeness of \(C^1(M)\), \(j_{1}^{(0,n)}(\mu)\) is differentiable in \(\mu\) and \(\dmu\, j_{1}^{(0,n)}(\mu) = c_1^{(0,n)}(\mu)\).
    Thus, \(\Iedged{d}(\beta,\mu,b)\) is differentiable and satisfies
    \begin{equation*}
        \bigl\lvert
            \dmu\, \Iedged{d}(\beta,\mu,b) - \dmu\, \Iedged[L]{d}(\beta,\mu,b)
        \bigr\rvert
        \leq
        \beta \, d \,C_J \, \xi(L-R)
        .
    \end{equation*}
    To take \(d\to \infty\), observe that by Proposition~\ref{prop:vanishing-bulk-current-derivative}
    \begin{equation*}
        \begin{aligned}
            \lvert \dmu\, j^{(0,n)}_1(\mu) \rvert
            &\leq
            \bigl\lvert
                \dmu\, j^{(0,n)}_1(\mu) - \dmu\, j^{(0,n)}_{1,L}(\mu)
            \bigr\rvert
            + \bigl\lvert
                \dmu\, j^{(0,n)}_{1,L}(\mu)
            \bigr\rvert
            \\&\leq
            \, \beta \, C_J \, \xi(L-R)
            + \Cm \, \beta \, \xi(n-D-R)
            ,
        \end{aligned}
    \end{equation*}
    which converges to \(\beta \, \Cm \, \xi(n-D-R)\) as \(L\to\infty\).
    Summation over \(n\) shows that \(\Iedge(\beta,\mu,b)\) is differentiable and satisfies~\eqref{eq:prop-localization-dmu-I}.
\end{proof}

Note that Proposition~\ref{prop:differentiability_I} together with the equality \(m(\beta,\mu,b)=\Iedge(\beta,\mu,b)\) from Theorem~\ref{thm:limit-I} proves Theorem~\ref{thm:muderivative}\@.
Additionally, we provide a bound on the difference of \(\dmu\, m_L(\beta,\mu,b)\) and \(\dmu\,\Iedge[L](\beta,\mu,b)\) in finite volume, which is analogous to the bound from Theorem~\ref{thm:main}\@.

\begin{proposition}\label{prop:dmu-m_L-equals-dmu-I_L}
    Let \(n \mapsto n^2 \, \zeta\supedge(n)\in \ell^1\), \((H_L(b))_{L\in\N}\) be a family of Hamiltonians of the form~\eqref{eq:Hb}.
    There exists a null sequence \(\eta\) so that the following holds: If \((H_L(b))_{L\in\N}\) satisfies local indistinguishability of the Gibbs state everywhere with \(\zeta\supedge\)-decay at \((\beta,\mu,b)\) for \(\mu\) in some open interval, then
    \begin{equation*}
        \Bigl\lvert
            \dmu\, m_L(\beta,\mu,b) - \dmu\, \Iedge[L](\beta,\mu,b)
        \Bigr\rvert
        \leq
        \beta \,\eta(L)
        .
    \end{equation*}
\end{proposition}
\begin{proof}
    Differentiating~\eqref{eq:magnetization-as-sum-over-currents} by using~\eqref{eq:mu-derivative-local-current}, we obtain
    \begin{equation*}
        \dmu \, m_L(\beta,\mu,b)
        =
        \frac{1}{(2L+1)^2} \sum_{m=-L}^{L-1} \sum_{n=0}^{2L} n \tr \bigl( \mathcal{F}_L(\beta,\mu,b) \, J_{1,L}^{(m,n)}(b) \bigr)
        .
    \end{equation*}
    Following the proof of Proposition~\ref{prop:m_L-equals-I_L^d}, where we only used \nameref{prop:Bloch-theorem} and current conservation in \(\rho_L(\beta,\mu,b)\), whose analogues here are Proposition~\ref{prop:vanishing-bulk-current-derivative} and current conservation in \(\mathcal{F}_L(\beta,\mu,b)\) (the latter holds, because \(\mathcal{F}_L(\beta,\mu,b)\) is stationary), we obtain
    \begin{equation*}
        \Bigl\lvert
            \dmu\, m_L(\beta,\mu,b) - \dmu\, \Iedged[L]{d}(\beta,\mu,b)
        \Bigr\rvert
        \leq
        \beta \,\Cm\, \biggl(
            \frac{4 \, d^2}{L} + \smashoperator[lr]{\sum_{n=d-R-D}^\infty} \xi(n)
        \biggr).
    \end{equation*}
    Combining this with Proposition~\ref{prop:vanishing-bulk-current-derivative} to compare \(\dmu\,\Iedged[L]{d}(\beta,\mu,b)\) and \(\dmu\,\Iedge[L](\beta,\mu,b)\), proves the claim with
    \begin{equation*}
        \eta(L)
        =
        2 \, \Cm \min_{D+R\leq d\leq L} \biggl(
            \frac{2 \, d^2}{L} + \smashoperator[lr]{\sum_{n=d-R-D}^\infty} \xi(n)
        \biggr).
        \qedhere
    \end{equation*}
\end{proof}

\subsection{Proof of Theorem~\ref{thm:limit-p}}
\label{sec:proof-theorem-limit-p}

We first prove that the limit of the finite volume pressures \(p_L(\mu,\beta,b)\) defined in~\eqref{eq:definition-pressure-finite-volume-partition-function} exists and is independent of the edge Hamiltonian.
Therefore, let
\begin{equation}
    \label{eq:definition-rho-bulk}
    \rho\supbulk_L(\beta,\mu,b)
    =
    \rho_{B_L}\bigl[H\supbulk_L(b)\bigr](\beta,\mu)
    ,
    \quadtext{and}
    \mathcal{Z}\supbulk_L
    =
    \mathcal{Z}_{B_L}\bigl[H\supbulk_L(b)\bigr](\beta,\mu)
\end{equation}
be the Gibbs state and partition function of the bulk Hamiltonian on the centered boxes~\(B_L\).
For the statement we introduce
\begin{equation}
    \label{eq:bound-local-Hamiltonian-terms}
    \CHany
    :=
    \sup_{x\in \Z^2}
    \biggl(
        2 \sum_{y\in \Z^2}
        \bigl\lVert a^*_x \, T\supany_b(x, y) \, a^{}_y \bigr\rVert
        + \sumstack{X\subset \Z^2:\\x\in X}
        \bigl\lVert \Phi\supany(X) \bigr\rVert
    \biggr)
    + \mu
    ,
\end{equation}
which bounds the norm of all hoppings and interactions which include a particular site.

\begin{proposition}
    \label{prop:convergence-pressure-independence-boundary}
    Let \(\bigl(H_L(b)\bigr)_{L\in\N}\) be a family of Hamiltonians of the form~\eqref{eq:Hb} and let \(H_L\supbulk(b)\) be the corresponding bulk Hamiltonian defined in~\eqref{eq:Hbulk}.
    Then
    \begin{equation}
        \label{eq:prop-bound-pressure-edge-and-bulk-system}
        \biggl\lvert
            p_L(\beta,\mu,b)
            - \frac{P\bigl(\rho\supbulk_L(\beta,\mu,b)\bigr)}{(2L+1)^2}
        \biggr\rvert
        \leq
        \frac{\CHedge \, D}{2L+1}
    \end{equation}
    for all \(\beta>0\), \(\mu\), \(b\in \R\).
    Moreover, the thermodynamic limit \(p(\beta,\mu,b) := \lim_{L\to\infty} p_L(\beta,\mu,b)\) of the pressure exists and
    \begin{equation}
        \label{eq:prop-convergence-pressure-bulk-system}
        \biggl\lvert
            p(\beta,\mu,b)
            - \frac{P\bigl(\rho\supbulk_L(\beta,\mu,b)\bigr)}{(2L+1)^2}
        \biggr\rvert
        \leq
        \frac{4\,R\,\CHbulk}{2L+1}
        .
    \end{equation}
\end{proposition}

The proof is based on~\cite[section~9.2]{AM2003}, where the convergence for translation invariant interactions is discussed.
Instead, here we have a bulk part which is invariant under magnetic translations and an additional edge contribution.

\begin{proof}
    Within the proof we fix \(\beta\), \(\mu\) and denote \(P(H) = \beta^{-1} \, \ln \tr \bigl(\E^{-\beta H} \bigr)\) for self-adjoint operators \(H\) such that \((2L+1)^2 \, p_L(\beta,\mu,b) = P\bigl(\rho_L(\beta,\mu,b)\bigr) = P\bigl(H_L(b)-\mu \, \mathcal{N}_L\bigr)\), i.e.\ we write the pressure of the Gibbs state by just specifying the exponent.

    We begin with the important observation, that the pressure is continuous and bounded in the Hamiltonian, i.e.\ \(\bigl\lvert P(H_1) - P(H_2) \bigr\rvert \leq \lVert H_1 - H_2 \rVert\) and \(\bigl\lvert P(H_1) \bigr\rvert \leq \lVert H_1 \rVert\), for all self-adjoint \(H_1\) and \(H_2\).
    To see this, consider self-adjoint \(A_1\) and \(A_2\) and \(A(\lambda) := \lambda \, A_1 + (1-\lambda) \, A_2\).
    Then
    \begin{equation*}
        \begin{aligned}
            \bigl\lvert
                \ln \tr(\E^{-A_1}) - \ln \tr(\E^{-A_2})
            \bigr\rvert
            &=
            \biggl\lvert
                \int_0^1 \frac{\D}{\D \lambda} \ln \tr \bigl(\E^{-A(\lambda)} \bigr)
                \, \D \lambda
            \biggr\rvert
            \\&=
            \biggl\lvert
                \int_0^1
                \frac
                    {\tr \bigl( (A_1-A_2) \, \E^{- A(\lambda)} \bigr)}
                    {\tr \bigl( \E^{- A(\lambda)} \bigr)}
                \, \D \lambda
            \biggr\rvert
            \leq
            \lVert A_1-A_2 \rVert
            ,
        \end{aligned}
    \end{equation*}
    where we used that \(\E^{- A(\lambda)} \big/ \tr \bigl(\E^{- A(\lambda)}\bigr)\) is a normalized state in the last step.
    The result for \(P\) follows immediately because the factors of \(\beta\) cancel.

    We first show that the pressure in finite volume is almost independent of the edge terms.
    Therefore, let
    \(
        W := \sum_{x,y \in \Lambda_L} a^*_x \, T\supedge_b(x, y) \, a^{}_y
        + \sum_{X\subset \Lambda_L} \Phi\supedge(X)
    \)
    be the edge contribution to the Hamiltonian \(H_L(b)\) such that \(H_L(b) - W\) is the Hamiltonian from~\eqref{eq:Hb} without any additional edge terms.
    Then,
    \begin{equation*}
        \bigl\lvert
            P\bigl(H_L - \mu \, \mathcal{N}_L\bigr)
            - P\bigl( H_L(b) - W  - \mu \, \mathcal{N}_L \bigr)
        \bigr\rvert
        \leq
        \lVert W \rVert
        \leq
        \CHedge \, D \, (2L+1)
        .
    \end{equation*}
    Thus, the per volume pressure \(p_L(\beta,\mu,b)\) is independent of the edge terms up to an error \(\CHedge \, D \big/ (2L+1) \to 0\) as \(L\to\infty\) and we only consider the Hamiltonian without edge terms in the following.

    To shorten notation in the following we denote \(H\supbulk_\Lambda := H\supbulk_L\big|_\Lambda\) for \(\Lambda\subset B_L\).
    As discussed in section~\ref{sec:magnetic-translations}, \(H\supbulk_{\Lambda+x} = U_{-x}(b) \, H\supbulk_{\Lambda} \, U_{-x}^*(b)\) and thus the partition function and the pressure of the respective states agree, \(P(H\supbulk_{\Lambda+x}) = P(H\supbulk_{\Lambda})\).
    This in particular proves that \(P\bigl(H_L(b)-\mu\,\mathcal{N}_{\Lambda_L}\bigr) = P\bigl(H\supbulk_L(b)-\mu\,\mathcal{N}_{B_L}\bigr)\), i.e.\ the pressure of the system on \(\Lambda_L\) without edge contribution exactly agrees with that on \(B_L\).
    Together with the above estimate, \eqref{eq:prop-bound-pressure-edge-and-bulk-system} follows.

    We now prove convergence of \(p\supbulk_L := P\bigl(\rho\supbulk_L(\beta,\mu,b)\bigr)\big/(2L+1)^2\) as \(L\to\infty\).
    For \(L'<L\) one can fit \(n = \bigl\lfloor \frac{2L+1}{2L'+1} \bigr\rfloor^2\) disjoint boxes \(B_{L'}+x_j\) in \(B_L\).
    By the estimate on the pressure, we find
    \begin{align*}
        \biggl|
            P\bigl(
                H\supbulk_{B_L}
                - \mu \, \mathcal{N}_{B_{L}}
            \bigr)
            -
            P\bigl(
                H\supbulk_{\bigcup_j B_{L'}+x_j}
                - \mu \, \mathcal{N}_{\bigcup_j B_{L'}+x_j}
            \bigr)
        \biggr|
        &\leq
        \CHbulk \, \bigl(\lvert B_L \rvert - n \, \lvert B_{L'} \rvert\bigr)
        \\&\leq
        2 \, \CHbulk \, (2L'+1) \, (2L+1)
        .
    \end{align*}
    In the second step we used \(\bigl\lvert\lfloor q \rfloor^2 - q^2 \bigr\rvert \leq 2 \, q\) for \(q>0\).
    In the next step we remove the hoppings and interactions between the individual boxes
    \begin{align*}
        &\alignindent
        \biggl|
            P\bigl(
                H\supbulk_{\bigcup_j B_{L'}+x_j}
                - \mu \, \mathcal{N}_{\bigcup_j B_{L'}+x_j}
            \bigr)
            - P\Bigl(
                \sum_{j=1}^n H\supbulk_{B_{L'}+x_j}
                - \mu \, \mathcal{N}_{B_{L'}+x_j}
            \Bigr)
        \biggr|
        \\&\leq
        4 \, \CHbulk \, R \, (2L'+1) \, n
        \\&\leq
        4 \, \CHbulk \, R \, (2L+1)^2 \, (2L'+1)^{-1} .
    \end{align*}
    Then, we observe that the trace of the non-interacting parts factors
    \begin{align*}
        P\Bigl(
            \sum_{j=1}^n H\supbulk_{B_{L'}+x_j} - \mu \, \mathcal{N}_{B_{L'}+x_j}
        \Bigr)
        &=
        \beta^{-1} \, \ln \prod_{j=1}^n \tr \bigl(\beta\,H\supbulk_{B_{L'}+x_j}-\beta\,\mu \, \mathcal{N}_{B_{L'}+x_j}\bigr)
        \\&=
        \sum_{j=1}^n
        P\Bigl(
            H\supbulk_{B_{L'}+x_j}
            - \mu \, \mathcal{N}_{B_{L'}+x_j}
        \Bigr)
        \\&=
        n \, P\bigl(
            H\supbulk_{B_{L'}}
            - \mu \, \mathcal{N}_{B_{L'}}
        \bigr)
        ,
    \end{align*}
    where we used that the pressures of the individual boxes all agree.
    As a last step we bound
    \begin{equation*}
        \biggl\lvert
            \biggl(
                n
                - \frac{\lvert B_{L} \rvert}{\lvert B_{L'} \rvert}
            \biggr)
            \, P\bigl(
                H\supbulk_{B_{L'}}
                - \mu \, \mathcal{N}_{B_{L'}}
            \bigr)
        \biggr\rvert
        \leq
        2 \, \CHbulk \, (2L'+1) \, (2L+1)
        .
    \end{equation*}
    Using triangle inequality and dividing everything by \(\lvert B_L \rvert\), we obtain an estimate for the per volume pressures
    \begin{equation}
        \label{eq:proof-p_L-bulk-is-Cauchy-sequence}
        \Bigl\lvert
            p\supbulk_L
            - p\supbulk_{L'}
        \Bigr\rvert
        \leq
        4 \, \CHbulk
        \, \biggl(
            \frac{2L'+1}{2L+1}
            + \frac{R}{2L'+1}
        \biggr)
        .
    \end{equation}
    Equation~\eqref{eq:proof-p_L-bulk-is-Cauchy-sequence} shows that \(\{p\supbulk_L\}_{L \in \N}\) is a Cauchy sequence and thus it is convergent.
    Together with~\eqref{eq:prop-bound-pressure-edge-and-bulk-system} also \(p_L(\beta,\mu,b)\) converges to the same limit \(p(\beta,\mu,b)\).
    The convergence in~\eqref{eq:prop-convergence-pressure-bulk-system} follows from~\eqref{eq:proof-p_L-bulk-is-Cauchy-sequence} after taking the limit \(L\to\infty\).
\end{proof}

Next we show that the limit of the pressures of the finite volume boxes agrees with the pressure of any infinite volume KMS state of the system without an edge.

\begin{proposition}
    \label{prop:limit-of-pressure-equals-infinitie-volume-pressure}
    Let \(\bigl(H\supbulk_L(b)\bigr)_{L\in \N}\) be a family of Hamiltonians of the form~\eqref{eq:Hbulk} satisfying the assumptions from section~\ref{sec:Hamiltonian}, and let \(\tau\supbulk\) be the corresponding dynamics defined in~\eqref{eq:dynamics-Hbulk}.
    For every \((\tau\supbulk,\beta)\)-KMS state \(\omega\) the pressure per volume of the restriction of \(\omega\) to \(B_L\) defined by~\eqref{eq:definition-pressure-state}, satisfies
    \begin{equation*}
        \biggl|
            \frac{P\bigl(\rho\supbulk_L(\beta,\mu,b)\bigr)}{(2L+1)^2}
            - \frac{P\bigl(\omega|_{B_L}\bigr)}{(2L+1)^2}
        \biggr|
        \leq
        \frac{8 \, \CHbulk \, R}{2L+1}
        .
    \end{equation*}
\end{proposition}

\begin{proof}
    We follow the ideas of~\cite[Proposition~12.1]{AM2003}.
    Let \(\omega\) be a KMS state, denote its restriction \(\omega_L:= \omega|_{B_L}\) and abbreviate \(\rho\supbulk_L = \rho\supbulk_L(\beta,\mu,b)\).
    The difference
    \begin{equation*}
        \beta \, P(\omega_L) - \beta \, P(\rho\supbulk_L)
        =
        \tr(\omega_L \ln \omega_L)- \tr(\omega_L \ln \rho\supbulk_L)
    \end{equation*}
    equals the relative entropy \(S(\omega_L | \rho\supbulk_L) \geq 0\).
    Since the relative entropy is monotone under restrictions (see~\cite[Theorem~6.2.33]{BR1997}), we have
    \begin{equation*}
        S(\omega_L | \rho\supbulk_L) \leq S(\omega|\rho)
    \end{equation*}
    for any extension \(\rho\) of \(\rho\supbulk_L\) from \(\alg_{B_L}\) to \(\alg\).
    By~\cite[Theorem~7.5]{AM2003}, \(\omega\) satisfies the Gibbs condition and a natural choice for this extension is given by the perturbation of~\(\omega\) where all interactions between \(B_L\) and the rest of the system are deleted (compare~\cite[Corollary~7.8]{AM2003}).
    To be precise, let
    \begin{equation*}
        W_L
        =
        \sumstack[lr]{
            x,y \in \Z^2 \colon
            \\ \{x,y\} \cap B_L \neq \emptyset,
            \\ \{x,y\}\cap \Z^2 \setminus B_L \neq \emptyset
        } a^*_x \, T\supbulk(x,y) \, a^{}_y
        + \sumstack[lr]{
            X \subset \Z^2 \colon
            \\ X \cap B_L \neq \emptyset,
            \\ X \cap \Z^2 \setminus B_L \neq \emptyset
        } \Phi\supbulk(X)
    \end{equation*}
    be the surface interaction, which is an element of \(\alg\) with norm bounded by \({4 \, \CHbulk \, R \, (2L+1)}\) since all interactions are of finite range~\(R\).
    The state corresponding to subtraction of \(W_L\) from the Hamiltonian can be expressed in the GNS representation \((\mathfrak{h}_{\omega}, \pi_{\omega}, \Omega)\) for \(\omega\) by
    \begin{equation*}
        \rho(A)
        =
        \Bigl\langle
            \E^{-\beta (H_\omega-\pi_\omega(W_L))/2}
            \, \E^{\beta H_\omega/2}
            \, \Omega,
            \pi_\omega(A)
            \, \E^{-\beta (H_\omega-\pi_\omega(W_L))/2}
            \, \E^{\beta H_\omega/2}
            \, \Omega
        \Bigr\rangle \Big/ \mathcal{Z}^{W_L},
    \end{equation*}
    where \(H_\omega\) is the generator of the dynamics induced by \(\tau\supbulk\) in \(\mathfrak{h}_\omega\) and \(\mathcal{Z}^{W_L}\) the normalizing factor (see~\cite[Theorem~5.4.4]{BR1997} and note that \(\E^{\beta H_\omega/2}\) acts trivially on the cyclic vector \(\Omega\) since \(\omega\) is invariant).
    With this, we have (cf.~\cite[below Definition~6.2.29]{BR1997})
    \begin{equation*}
        S(\omega|\rho)
        \leq
        -\omega(\beta \, W_L) + \rho(\beta \, W_L)
        \leq
        2 \, \beta \, \lVert W_L \rVert,
    \end{equation*}
    and thus
    \begin{equation*}
        0
        \leq
        \frac{ P(\rho\supbulk_L) + P(\omega_L)}{(2L+1)^2}
        \leq
        \frac{8 \, \CHbulk\, R}{2L+1}
        .
        \qedhere
    \end{equation*}
\end{proof}

Now we are able to prove Theorem~\ref{thm:limit-p}.

\begin{proof}[Proof of Theorem~\ref{thm:limit-p}]
    The convergence of \(p_L(\beta,\mu,b)\) and independence of boundary terms follows from Proposition~\ref{prop:convergence-pressure-independence-boundary}.
    Equality of the pressure with the per volume pressure of any \((\tau\supbulk,\beta)\)-KMS state follows from Propositions~\ref{prop:convergence-pressure-independence-boundary} and~\ref{prop:limit-of-pressure-equals-infinitie-volume-pressure}.

    Now assume that \(\bigl(H_L(b)\bigr)_{L\in \N}\) satisfies local indistinguishability uniformly in \(b\).
    By Theorem~\ref{thm:limit-I}, \(m_L(\beta,\mu,b) = \partial_b \, p_L(\beta,\mu,b)\) converges to \(m(\beta,\mu,b)\) as \(L\to \infty\), and in view of the estimate~\eqref{eq:proof-explicit-bound-m_L-equals-I} this convergence is uniform in \(b\).
    Then, the convergence of the primitives \(p_L(\beta,\mu,b)\), which converge pointwise by Proposition~\ref{prop:convergence-pressure-independence-boundary}, must also be uniform and \(p(\beta,\mu, b)\) is differentiable with derivative \(m(\beta,\mu,b)\).
\end{proof}

\statement{Acknowledgments}

We thank Sven Bachmann for valuable comments on the first version of the manuscript.
This work was funded by the \foreignlanguage{ngerman}{Deutsche Forschungsgemeinschaft} (DFG, German Research Foundation) –
470903074; 
465199066
.
The work of M.\,M. has been supported by a fellowship of the Alexander von Humboldt Foundation during his stay at the University of Tübingen, where this work initiated.
M.\,M. gratefully acknowledges the support of \foreignlanguage{italian}{PNRR Italia Domani} and Next Generation EU through the \LTskip{ICSC} National Research Centre for High Performance Computing, Big Data and Quantum Computing and the support of the MUR grant \foreignlanguage{italian}{Dipartimento di Eccellenza}~2023–2027.
J.\,L. acknowledges additional support through the EUR-EIPHI Graduate School (ANR-17-EURE-0002).

\statement{Competing interests}
The authors have no relevant financial or non-financial interests to disclose.

\statement{Data availability}
No datasets were generated or analyzed during the current study.

\printbibliography[heading=bibintoc]

\end{document}